\documentclass{amsart}

\input{epsf.sty}

\usepackage{amsmath,amssymb}
\usepackage{graphicx,color,pstricks}
\newtheorem{theorem}{Theorem}[section]
\newtheorem{lemma}[theorem]{Lemma}
\newtheorem{proposition}[theorem]{Proposition}

\theoremstyle{definition}
\newtheorem{definition}[theorem]{Definition}
\newtheorem{example}[theorem]{Example}

\theoremstyle{remark}
\newtheorem{remark}[theorem]{Remark}

\newcounter{num}
\renewcommand{\thenum}{\roman{num}}

\newenvironment{hlist}
{\begin{list}{}{
\setlength{\leftmargin}{2em}
\setlength{\labelwidth}{1.6em}
\setlength{\itemsep}{0.1ex}
}}{\end{list}}

\newtheorem{corollary}[theorem]{Corollary}

\numberwithin{equation}{section}

\newcommand\la{\lambda}
\newcommand\al{\alpha}
\newcommand\de{\delta}

\newcommand\CL{{\mathcal L}}
\newcommand\CH{{\mathcal H}}

\newcommand\CA{{\mathcal A}}
\newcommand\CB{{\mathcal B}}
\newcommand\CC{{\mathcal C}}
\newcommand\CD{{\mathcal D}}
\newcommand\CQ{{\mathcal Q}}
\newcommand\CR{{\mathcal R}}

\newcommand\CS{{\mathcal S}}
\newcommand\CT{{\mathcal T}}
\newcommand\CW{{\mathcal W}}
\newcommand{\efrac}[2]{\genfrac{}{}{0pt}{}{#1}{#2}}
\renewcommand{\matrix}[2]{\left( \!\! \begin{array}{#1} #2 \end{array} \!\! \right)}

\newcommand{\ov}[1]{\overline{#1}}
\renewcommand\d{\,{\rm d}}

\newcommand{\N}{{\mathbb N}}

\newcommand{\R}{{\mathbb R}}
\newcommand{\C}{{\mathbb C}}

\newcommand\eps{\varepsilon}

\newcommand{\ds}{\displaystyle}

\renewcommand\Re{{\rm Re\,}}
\renewcommand\Im{{\rm Im\,}}
\newcommand\dist{{\rm dist\,}}

\newcounter{marke}
\newcommand{\bl}{\begin{list}{\roman{marke})}{\usecounter{marke}
\topsep 0 cm \itemsep 0cm}}
\newcommand{\el}{\end{list}}

\begin{document}

\title[Spectrum of magnetohydrodynamic mean-field dynamo operators]
{On the spectrum of the magnetohydrodynamic mean-field {\boldmath $\alpha^2$}-dynamo operator}

% Remove or comment out any unused author tags.
% author one information
\author{Uwe G\"unther}
\address{Research Center Dresden-Rossendorf, P.O. Box 510119, D-01314 Dresden, Germany}
\email{u.guenther@fzd.de}
\author{Heinz Langer}
\address{Institut f\"ur Analysis und Scientific Computing, Technische Universit\"at
Wien, Wiedner Hauptstr.\ 8--10, A--1040 Wien, Austria}
\email{hlanger@email.tuwien.ac.at}
\author{Christiane Tretter}
\address{Mathematisches Institut, Universit\"at Bern, Sidlerstr.\ 5, 3012 Bern, Switzerland}
\email{tretter@math.unibe.ch}
\thanks{The first and the last author gratefully acknowledge support by the Deutsche Forschungsgemeinschaft, DFG, within the Collaborative Research Center (Sonderforschungsbereich) 609 and under Grant No.\ TR368/6-1, respectively.}

\date{01.04.2010}

\begin{abstract}
The existence of magnetohydrodynamic mean-field $\al^2$-dynamos
with spherically symmetric, isotropic helical turbulence function $\al$ is related to 
a non-self-adjoint spectral problem for a coupled system of two singular second order 
ordinary differential equations. We establish global estimates for the 
eigenvalues of this system in terms of the turbulence function $\al$ and its derivative~$\al'$.
They allow us to formulate an anti-dynamo theorem and a non-oscillation theorem.
The conditions of these theorems, which again involve $\al$ and $\al'$, must be violated 
in order to reach supercritical or oscillatory regimes.
\end{abstract}

\maketitle

\section{Introduction}
\label{intro}

One of the simplest models of a magnetohydrodynamic dynamo is a mean-field $\alpha^2$-dynamo with spherically symmetric, isotropic turbulence
function ($\alpha$-profile) $\alpha(r)$. Models of this type were
among the first that were capable to explain the turbulence-based
inverse energy cascade leading to non-decaying dynamo regimes with
self-sustaining magnetic fields (see \cite{MR668520}). Recently, $\alpha^2$-dynamo models~were used for detailed studies of polarity reversal processes of
magnetic fields (see,~e.g., \cite{stefani:184506}, \cite{2006E&PSL.243..828S}),
as suggested by paleomagnetic data of the Earth's magnetic field.
Deep insight into the reversal dynamics could be gained by analyzing the highly non-linear
back-reaction-based induction processes by a time foliation method built
over a series of instantaneously linearized (kinematic) auxiliary setups. For the
latter, the $\alpha$-profiles were assumed to be fixed radial functions
so that the well-known decomposition of the magnetic field into poloidal and toroidal components 
with subsequent expansion in spherical harmonics could be used. As a result, one arrives at a set 
of eigenvalue problems, indexed by the degree $l\in\N=\{1,2,\dots\}$ of the spherical harmonics, 
for pairs of coupled linear ordinary differential~equations %(see \cite{MR668520}, \cite{MR1982771}) 
\begin{align}
\label{diffsyst}
 \matrix{cc}{ \partial_r^2 - \displaystyle{\frac{l(l+1)}{r^2}} & \alpha(r) \\ 
              - \partial_r \alpha(r) \partial_r + \alpha(r)\displaystyle{\frac{l(l+1)}{r^2}} & \partial_r^2 - \displaystyle{\frac{l(l+1)}{r^2}}}
 \binom{y_1}{y_2} = \lambda \binom{y_1}{y_2}, \quad r\in (0,1],
\end{align}
in $L_2(0,1) \oplus L_2(0,1)$ subject to the boundary condition 
\begin{align}
\label{boundcond}
 % \binom{y_1}{y_2}(0) = 0, 
 \quad \binom{(\partial_r + l ) y_1}{y_2}(1) = 0
\end{align}
(see \cite[Section~14.2]{MR668520}, \cite[(6)--(8)]{stefani-2003-67}). 
In the following, for brevity, we call the eigenvalue problem \eqref{diffsyst}, \eqref{boundcond} dynamo problem. 

In view of the time separation ansatz ${\bf \bar B}(x,t) = \exp(\lambda t) {\bf \bar b}(x)$ of the corresponding magnetic field modes, 
one naturally distinguishes between decaying or subcritical modes ($\Re \lambda<0$) and amplifying or supercritical modes ($\Re \lambda>0$), as well as between
oscillatory modes ($\Im \lambda\neq 0$) and non-oscillatory modes ($\Im\lambda=0$). The physically relevant self-sustaining dynamo configurations are
mainly defined by a few supercritical modes, whereas possible
polarity reversals of the magnetic fields are closely related to the
existence of oscillatory modes close to criticality ($\Re \lambda \approx 0$, $\Im \la \ne 0$) (see~\cite{stefani:184506}, \cite{2006E&PSL.243..828S}). 
For a deeper understanding of the physical interplay between the plasma- (or conducting fluid-) based $\alpha$-profile and the dynamo
dynamics, knowledge of the spectral properties of the dynamo problem \eqref{diffsyst}, \eqref{boundcond} 
is of utmost~interest. %In the present work we are going to fill this still existing gap.

{
\renewcommand{\i}{{\rm i}}

% One of the simplest models of a magnetohydrodynamic (MHD) dynamo is a mean-field $\alpha^2$-dynamo with 
% spherically symmetric isotropic, helical turbulence parameter $\alpha$. In the kinematic mean-field 
% dynamo problem one considers magnetic mean-fields of the form $\mathbf{\bar B}({\bf x},t) = \exp(\la t) \, \mathbf{\bar B}({\bf x})$.
% The problem of existence of an oscillatory $\alpha^2$-dynamo amounts to showing that the corresponding linearized spectral problem 
% has an eigenvalue $\la_0$ with $\Re \la_0 >0$ and $\Im \la_0 \ne 0$.
% %$\la = \sigma + \i \omega$ with positive real part $\sigma$ and non-zero imaginary part~$\omega$. 
% This linearized spectral problem is described by a system of two coupled second order ordinary differential equations 

In contrast to the large number of general results for kinematic fast dynamo problems (see, e.g., \cite[Chapter~V]{MR1612569} as well as \cite{MR1355630}) and despite the simplicity of the physical model, no proven analytic information for the spectrum of the kinematic mean-field dynamo problem  \eqref{diffsyst}, \eqref{boundcond} seems to be available, except for the case of constant $\alpha$ where the spectrum is real. For non-constant functions $\alpha$, only numerical calculations for eigenvalues were performed (see, e.g., \cite{stefani-2003-67}). However, for a non-symmetric spectral problem like \eqref{diffsyst}, \eqref{boundcond}, numerical computations are prone to be unreliable; a con\-vincing example for this is a non-normal $7\times 7$ matrix due to S.K.~Godunov for which numerical algorithms yield results far away from the true eigenvalues (see, e.g., \cite{MR1620203}, \cite[Example~5.2.5]{MR2116013}).

Other attempts to attack the above dynamo problem include replacing the physical boundary conditions by the so-called idealized boundary condition
\begin{align}
\label{boundcond-id}
 %\binom{y_1}{y_2} (0) = 0, \quad 
 \binom{y_1}{y_2}(1) = 0,
 \end{align}
assuming that the helical turbulence function $\alpha$ is strictly positive, and/or analyzing the simpler monopole modes, that is, the case $l=0$, 
(see, e.g., \cite{MR1982771}, \cite{MR2252699}, \cite{kirillov:016205}); all these restrictions are not desirable from the physical point of view.

The aim of this paper is to establish analytic enclosures for the spectrum of the dynamo problem \eqref{diffsyst}, \eqref{boundcond} 
in terms of the function $\alpha$, not making any of the above simplifying assumptions. We show, in particular, that the eigenvalues lie in a horizontal strip around the real axis and we establish an upper bound for their real parts. More precisely, every eigenvalue $\la$ of \eqref{diffsyst}, \eqref{boundcond} satisfies
\begin{equation}
\label{optimal}
  |\Im \la| \le \|\al'\|, \quad \Re \la \le \rho_\theta
\end{equation}
where the constant $\rho_\theta$ depends on $\|\al\|$ and $\|\al'\|$, as well as  
on the smallest eigenvalues $\la_1(l)$ and $\la_1(\infty)$ of the Bessel operator $-\partial_r^2 + l(l+1)/r^2$ with boundary condition
$y'(1)+ly(1)=0$ and $y(1)=0$, respectively (compare \eqref{boundcond}). Here, for a continuous function $q$  on $[0,1]$, we denote by $\|q\|:=\max_{r\in[0,1]}|q(r)|$ the maximum norm of $q$.
The estimate for the real part in \eqref{optimal} yields, in particular, a so-called anti-dynamo theorem: if 
%as well as a non-oscillation theorem. The anti-dynamo theorem 
\[
 \|\alpha\|^2 + \dfrac{\|\al\|\|\alpha'\|}{\sqrt{\lambda_1(l)}} < \lambda_1(\infty), 
\]
%which states criteria on $\|\al\|$ and $\|\al'\|$ preventing 
then the dynamo operator has no eigenvalues in the (supercritical) right half-plane (compare the more geometrical anti-dynamo theorems for kinematic fast dynamos listed, e.g., in \cite[Chapter V, \S 3]{MR1612569}).
%, whereas the non-oscillation theorem states similar criteria for the eigenvalues to remain on the real axis (non-oscillatory) and in this way to prevent possible polarity reversals. 
In addition, we establish conditions %in terms of $\|\al\|$ and $\|\al'\|$ 
ensuring that a particular eigenvalue remains on the real axis when $\alpha \not\equiv 0$, which could be called a local non-oscillation theorem.
Our results also cover other boundary conditions such as the idealized case \eqref{boundcond-id}, and they do not require $\alpha$ to be sign definite. 
The methods we employ stem from the perturbation theory of linear operators and from the spectral theory of block operator matrices 
(see, e.g., \cite{MR0407617},~\cite{book}). 

In the following we give a brief outline of the paper.
In Section \ref{section2} we introduce a family of linear operators $\CA_\theta$ in $L_2(0,1) \oplus L_2(0,1)$
so that we can write the boundary eigenvalue problems \eqref{diffsyst}, \eqref{boundcond} and \eqref{diffsyst}, \eqref{boundcond-id} as 
spectral problems $(\CA_\theta-\la) y = 0$ with $\theta=l$ and $\theta=\infty$, respectively. The operators $\CA_\theta$ are block operator matrices of the form
\[
 \CA_\theta := \matrix{cc}{ -A_\theta & \alpha \\ \hspace{4.5mm} A_{\theta,\alpha} & - A_\infty}, \quad 
 \CD(\CA_\theta) := \CD(A_\theta) \oplus \CD(A_\infty).
\]
Here $A_\theta$ and $A_{\theta,\al}$ are linear operators in $L_2(0,1)$ given by the Bessel (type) differential expressions $-\partial_r^2 + l(l+1)/r^2$ 
and $-\partial_r \al \partial_r +\al  l(l+1)/r^2$, respectively, and the boundary condition $x'(1)+\theta x(1)=0$ for $\theta \in [0,\infty]$ 
(note that $\theta=\infty$ corresponds to $x(1)=0$).
In Section \ref{sec-aux} we investigate the entries of the block operator matrices $\CA_\theta$ in detail 
and collect all the properties needed in the subsequent sections, e.g., to show that $\CA_\theta$ defines a closed linear operator and that its spectrum consists only of eigenvalues 
of finite algebraic multiplicities with no finite accumulation point.

Section \ref{section3} contains our first main result, the eigenvalue enclosure in Theorem~\ref{!!!}. 
This first estimate is based on decomposing $\CA_\theta$ into a lower triangular block operator matrix 
plus a bounded part containing only the right upper corner $\al$:
\[
   \CA_\theta := \matrix{cc}{ -A_\theta & 0 \\ \hspace{4.5mm} A_{\theta,\alpha} & - A_\infty}  + \matrix{cc}{0 & \al \\ 0 & 0}=\CQ_\theta + \CR.
\]
This decomposition, together with a Neumann series argument, enables us to identify a region in the complex plane in which the eigenvalues must lie. In addition,
we establish bounds on $\|\alpha\|$ and $\|\alpha'\|$ such that an eigenvalue of one of the diagonal elements of $\CA_\theta$, that is, an eigenvalue when $\alpha\equiv 0$, remains real 
for $\alpha \not\equiv 0$ (see Propositions \ref{local} and \ref{question}). 
Physically speaking, this means that a particular mode remains non-oscillating.
%, which have the physical meaning of spectrally-local non-oscillation theorems).

Section \ref{section4} contains our second main result, the eigenvalue enclosure in Theorem~\ref{strip}. 
Here we use a quasi-similarity transformation of $\CA_\theta$ (with the unbounded operator $\CW_\theta = {\rm diag}\,(A_\theta^{1/2},I)$) 
such that the transformed operator $\CB_\theta$ is a bounded perturbation of a self-adjoint operator: 
\[
 \CB_\theta = \matrix{cc}{-A_\theta & A_\theta^{1/2} \alpha  \\  \alpha A_\theta^{1/2} & - A_\infty}
              + \matrix{cc}{0 & 0 \\ -\alpha' D A_\theta^{-1/2} & 0} = \CS_\theta + \CT_\theta. 
\]
This allows us to conclude that the eigenvalues of $\CA_\theta$ lie in discs of radius $\|\al'\|$ (the norm of the perturbation $\CT_\theta$) 
around the eigenvalues of the self-adjoint operator~$\CS_\theta$; in particular, the imaginary parts of the eigenvalues of $\CA_\theta$ are bounded by $\|\al'\|$ 
and their real parts are bounded by $\max \sigma(\CS_\theta) + \|\al'\|$.

In Section \ref{section5} we compare the two different eigenvalue enclosures of Sections~\ref{section3} and~\ref{section4}. It turns out that, generically, a combination of 
both estimates yields the best result: it consists of the uniform estimate for the imaginary parts of the eigenvalues from Section \ref{section4} and of the upper bound for the real parts of the eigenvalues from Section \ref{section3} (see \eqref{optimal} and Figure~\ref{stab}). At the end of Section~\ref{section5} we summarize the physically most relevant result in the form of an anti-dynamo~theorem (see \eqref{stable1}).  
%and of a non-oscillation theorem (see \eqref{}). 

In Section \ref{section-last} we illustrate our results by some examples. They include the case of constant $\alpha$ where the eigenvalues of the physical dynamo problem are only given  implicitly as solutions of an equation involving four different Bessel functions. We also consider the particular non-constant function $\alpha$ for which dipole-dominated oscillatory criticality was first found numerically by F.\ Stefani and G. Gerbeth (see \cite{stefani-2003-67});
in their computations they obtained special $\alpha$-profiles such that the right-most eigenvalues pass from $\Re \lambda<0$ to $\Re \lambda>0$ with $\Im \lambda\neq 0$
first for the dipole modes ($l=1)$ and only afterwards for quadrupole and higher-degree modes ($l>1$).

The following notation is used throughout the paper. In the Hilbert space~$L_2(0,1)$ the scalar product and norm 
are denoted by $(\cdot,\cdot)$ and $\|\cdot\|$, respectively. By ${\rm AC_{loc}}(I)$ we denote 
the space of locally absolutely continuous functions on a (sub-)\,interval $I\subset[0,1]$ (that is, the space
of functions that are differentiable (Lebesgue-) almost every\-where on~$I$ and have locally Lebesgue-integrable derivative on $I$). 
For a linear operator $T$  we denote by $\CD(T)$ its domain, by $\sigma(T), \,\sigma_{\rm p}(T),\, \rho(T)$ its spectrum, 
the set of its eigenvalues, and its resolvent set, respectively; if $T$ is bounded, we denote by $\|T\|$ its operator norm.
For a continuous function $q \in C([0,1])$, we use the same symbol $q$ to denote the (bounded) multiplication operator by the function $q$ in $L_2(0,1)$; in this case,
the operator norm is given by the maximum norm $\|q\|=\max_{r\in [0,1]} |q(r)|$. For details on linear operators in Hilbert spaces we refer the reader to 
\cite[Chapter~III, \S\ 2, 5, 6]{MR0407617}, \cite[Chapters~VI, VIII]{MR751959}, and~\cite{MR0493421}.
%If $\mu \in L_\infty(0,1)$ is an essentially bounded function, then $\| \mu \|_\infty := {\rm ess\,sup}_{x\in (0,1)} |\mu(x)|$.

Finally, we would like to thank Frank Stefani and Gunter Gerbeth (Research Center Dresden-Rossendorf) 
for useful discussions on the physics of dynamos and for sharing their numerical data with us.
Special thanks also go to Markus Wagenhofer (now Psylock, Regensburg) who performed various numerical calculations
while this manuscript was being prepared.

% In general, if $\CF$ is a function on $\C$ the values of which are closed linear operators from a Banach space $E$ into a Banach space $F$, then the spectrum, point spectrum etc.\ of $\CF$ are defined as
% \[
%   \sigma(\CF) := \{ \la \in \C : 0 \in \sigma(\CF(\la))\}, \quad 
%   \sigma_{\rm p}(\CF) := \{ \la \in \C : 0 \in \sigma_{\rm p}(\CF(\la))\}.
% \]
% The numerical range $W(\CF)$ of $\CF$ is defined as
% \[
%   W(\CF) := \{ \la\in\C : \bigl( \CF(\la) x,x \bigr) = 0 \text{ for some } x \in \CD(\CF(\la)), \, x\ne 0 \bigr\}.
% \]
% Obviously, $\sigma_{\rm p}(\CF) \subset W(\CF)$. Note that for the particular case that $\CF(\la) = T-\la$ with a linear operator $T$ in a Banach space $E$, 
% the above definitions coincide with those for linear operators.

\section{Operator model}
\label{section2}

In this section we introduce  a family of linear operators $\CA_\theta$ in the product Hilbert space $L_2(0,1)\oplus L_2(0,1)$ 
which describe the system of differential equations \eqref{diffsyst} with physical boundary conditions \eqref{boundcond} and 
idealized boundary conditions \eqref{boundcond-id}.

To this end, we first associate linear operators in the Hilbert space $L_2(0,1)$ with the two differential expressions 
\begin{align} 
\label{taus}
\tau := -\partial_r^2 + \displaystyle{\frac{l(l+1)}{r^2}}, \quad 
\tau_\alpha:= - \partial_r \alpha(r) \partial_r + \alpha(r)\displaystyle{\frac{l(l+1)}{r^2}},
\quad \ \partial_r:=\dfrac{{\rm d}}{{\rm d}r},
\end{align}
occurring in \eqref{diffsyst}. Here and in the following, we always suppose that $l\in\N=\{1,2,\dots\}$ is fixed and that $\al\!:[0,1]\to\R$ is a continuously differentiable real-valued function, $\alpha\in C^1([0,1])$. 

The classical Bessel differential expression $\tau$ appears twice on the diagonal in \eqref{diffsyst}, but it is subject to two different boundary conditions at $r=1$ in \eqref{boundcond}. Therefore we define a family of operators $A_\theta$, $\theta \in [0,\infty]$, by means of the boundary condition $y'(1) + \theta \, y(1) = 0$, so that $\theta=l$ yields the first order boundary condition for $y_1$ and $\theta=\infty$ yields the Dirichlet boundary condition for $y_2$ in \eqref{boundcond}.

\begin{definition}
\label{A_theta-def}
Let $\theta\in[0,\infty]$. Define the linear operator $A_\theta$ in $L_2(0,1)$ by
\begin{align*}
  \CD(A_\theta ) &:=\big \{ x \in L_2(0,1): x, x' \!\in {\rm AC_{loc}}((0,1]), \ \tau x \in L_2(0,1),  
                                        \ x'(1) + \theta \, x(1) = 0\big\},\\
      (A_\theta x )(r)&:= (\tau x) (r)= - x''(r) + \frac{l(l+1)}{r^2} x(r) \quad \text{for almost every }r\in (0,1].
\end{align*}
\end{definition}

\begin{remark}
\label{0!}
Every $x\in \CD(A_\theta)$ automatically satisfies $\lim_{r\searrow 0} x(r) =0$. This follows easily if one calculates the inverse $A_\theta^{-1}$ 
and checks that $\lim_{r\searrow 0}(A_\theta^{-1} f)(r)=0$ for $f\in L_2(0,1)$ (see, e.g., \cite[Lemma 2.6 and its proof]{MR2356213}); note that a fundamental system of $\tau y = 0$ is given by
the simple functions $x_1(r)=r^{-l}\!$, $x_2(r)=r^{l+1}\!$,~$r\!\in\!(0,1]$.
\end{remark}

The differential expression $\tau_\alpha$ in \eqref{taus} is not classical since its leading coefficient $\alpha$ may change sign and may hence also have singularities in the interior of the interval~$[0,1]$.
Since $\tau_\alpha$ is only located in an off-diagonal corner in \eqref{diffsyst}, the following definition is sufficient for us.

\begin{definition}
\label{A_theta-alpha-def}
Let $\theta\in[0,\infty]$. Introduce the linear operator $A_{\theta,\alpha}$ in $L_2(0,1)$~by
\begin{align*}
  \CD(A_{\theta,\alpha} ) &:= \CD(A_\theta), \\ %&:= \{ x \in L_2(0,1): x, x' \!\in {\rm AC_{loc}}((0,1]), \ \tau x \in L_2(0,1),  
                                                 % \ x'(1) + \theta \, x(1) = 0\}      
  (A_{\theta,\alpha} x)(r) &:= (\tau_\alpha x)(r) = - (\alpha x')'(r) + \alpha(r) \frac{l(l+1)}{r^2} x(r) \ \, \text{for  almost every }r\in (0,1].
\end{align*}
\end{definition}

\begin{remark}
In Proposition \ref{A_theta-alpha} below we show that $A_{\theta,\alpha} x \in L_2(0,1)$ for $x\in \CD(A_\theta)$, that is, $A_{\theta,\alpha}$ is well-defined.
\end{remark}

Now we are able to formulate the dynamo problems \eqref{diffsyst}, \eqref{boundcond} and \eqref{diffsyst}, \eqref{boundcond-id} 
as spectral problems for a linear operator $\CA_\theta$ acting in the product Hilbert space $L_2(0,1) \oplus L_2(0,1)$. Here
the case $\theta=l$ corresponds to the physical boundary condition \eqref{boundcond}, while $\theta=\infty$ corresponds 
to the idealized boundary condition~\eqref{boundcond-id}.

\begin{proposition}
\label{equiv}
Let $\theta\in[0,\infty]$. Define a linear operator  $\CA_\theta$ in $L_2(0,1) \oplus L_2(0,1)$ by the block operator matrix 
\begin{align}
\label{bomA_l}
 \CA_\theta := \matrix{cc}{ -A_\theta & \alpha \\ \hspace{4.5mm} A_{\theta,\alpha} & - A_\infty}, \quad 
 \CD(\CA_\theta) := \CD(A_\theta) \oplus \CD(A_\infty).
\end{align}
Then the boundary eigenvalue problems \eqref{diffsyst}, \eqref{boundcond} and \eqref{diffsyst}, \eqref{boundcond-id} 
can be written equivalently as 
$$
  (\CA_\theta-\la) y = 0, \quad y\in \CD(\CA_\theta),
$$ 
for $\theta=l$ and $\theta=\infty$, respectively.
\end{proposition}

\begin{proof}
The claim is evident from the definitions of the operator $\CA_\theta$ and of its entries $A_\theta$ and $A_{\theta,\al}$ for $\theta=l$ and $\theta=\infty$.
% The claim is clear for the system of differential equations \eqref{diffsyst} and 
% for the boundary conditions at $r=1$ in \eqref{boundcond} and \eqref{boundcond-id}.
% For the boundary condition at $r=0$ imposed in \eqref{boundcond} and \eqref{boundcond-id},
% it remains to be observed that the condition $y = (y_1,y_2)^{\rm t} \in \CD(\CA_\theta)$ automatically implies that %both components $y_1$, $y_2$ belong to $L_2(0,1)$ and hence
% $\lim_{r\searrow 0} y_1(r) = \lim_{r\searrow 0} y_2(r) =0$ by Remark \ref{0!}.
\end{proof}

\section{Auxiliary results}
\label{sec-aux}

In this section we study the properties of the entries $A_\theta$ and $A_{\theta,\al}$ of the block operator matrix~$\CA_\theta$ introduced in the previous section. They are used in the next two sections to derive estimates for the eigenvalues of the dynamo problem \eqref{diffsyst},~\eqref{boundcond}.

\begin{proposition}
\label{A_theta}
Let $\theta\in[0,\infty]$. The linear operator $A_\theta$ in Definition {\rm \ref{A_theta-def}} has the following properties: \vspace{1mm}
\begin{enumerate}
\item[{\rm i)}] $\!A_\theta$ is self-adjoint and positive; the domain of its square root is given by
\begin{align}\label{ob1}
\quad \CD\big(A_\theta^{1/2}\big)\!&=\!\Big\{ x\in L_2(0,1)\!:x \!\in\!{\rm AC_{loc}}((0,1]),\,  x', \dfrac xr\!\in\! L_2(0,1) \Big\}, \quad \theta \in [0,\infty),\\[2mm]
\quad \CD\big(A_\infty^{1/2}\big) \!&= \!\Big\{ x\in L_2(0,1)\!:x \!\in\!{\rm AC_{loc}}((0,1]),\,  x', \dfrac xr\!\in\! L_2(0,1),\,x(1) = 0 \Big\}.\label{ob2}
\end{align}
\item[{\rm ii)}] $\!A_\theta$ has compact resolvent, and the spectrum $\sigma(A_\theta)$ consists of a sequence of simple eigenvalues 
$0< \lambda_1(\theta) < \lambda_2(\theta) < \cdots$ tending to $\infty$.  \vspace{0.5mm}
\item[{\rm iii)}] $\!\la_k(\theta)$ is strictly increasing in $\theta$ for every $k=1,2,\dots$, and for different values $\theta_1$, $\theta_2$
the sequences $\big(\la_k(\theta_1)\big)_{k=1}^\infty$, $\big(\la_k(\theta_2)\big)_{k=1}^\infty$ interlace; in~particular,
\begin{equation}
\label{interlace}
  0 < \la_1(\theta) < \la_1(\infty) < \la_2(\theta) < \la_2(\infty) < \cdots, \quad \theta \in [0,\infty).
\end{equation}
\item[{\rm iv)}] For $0\le \theta_1 \le \theta_2 \le \infty$ and $\la\ge 0$, we have
\[
  (A_{\theta_2}+\la)^{-1} \le (A_{\theta_1}+\la)^{-1}.
\]
\end{enumerate}
\end{proposition}

} % end of newcommand \i because of citation Naimark

\begin{proof}
All claims of the proposition are well-known; for convenience of the reader, we repeat the main arguments.

i) The Bessel differential expression $\tau$ (with $l\ge 1$) is in limit point case at the singular end-point~$0$ (see, e.g., \cite[Appendix II, Section 9.IV]{MR1255973}); in fact, the $L_2(0,1)$-solutions of the differential equation $(\tau - \la) x = 0$ are spanned by the Riccati-Bessel function
(see \cite[10.3.1]{abramowitz+stegun})
\begin{equation}
\label{RiccBessel}
   f_l(r,\la):= r \sqrt{\la} \, j_l (r\sqrt{\la}) = \sqrt{\frac \pi 2} \sqrt{r\sqrt{\la}} \, J_{l+1/2}(r\sqrt{\la}), \quad r\in[0,1].
\end{equation}
This implies that the operator $A_\theta$ 
as defined above is self-adjoint (see \cite[Satz~13.21\,a)]{MR2382320}); moreover, the set 
\begin{equation}
\label{supp-comp}
 \CD_0(A_\theta) := \{ x\in \CD(A_\theta) : \text{supp\,$x$ compact in } (0,1] \}
\end{equation}
forms a core of $A_\theta$, that is, the closure of the restriction $A_\theta|_{\CD_0(A_\theta)}$ coincides with~$A_\theta$ (see \cite[Section~17.4]{MR0262880}). %; in particular, $x/r \in L_2(0,1)$ for every $x\in \CD(A_\theta)$.

The positivity of $A_\theta$ follows from the identity
\begin{align}
\label{ee}
  {\mathfrak a}_\theta[x]:=\big(A_\theta x,x\big) =:\langle x,x\rangle_\theta,\quad x\in \CD(A_\theta),
\end{align}
where
\begin{equation}
\label{inner}
  \langle x,y\rangle_\theta:=c_\theta \,x(1)\overline{y(1)}+\int_0^1x'(r)\ov{y'(r)}\d r+ l(l+1)\int_0^1 \frac{x(r)\ov{y(r)}}{r^2} \d r
\end{equation}
with $c_\theta=\theta$ for $\theta\in [0,\infty)$ and $c_\infty=0$. The inner product $\langle x,y\rangle_\theta$
is defined for all  $x,y\in L_2(0,1)$ such that $x'$, $y'$, $x/r$, $y/r \in L_2(0,1)$; in particular, it 
is defined on the sets on the right-hand sides of \eqref{ob1} and of \eqref{ob2}.

The domain  $\CD(A_\theta^{1/2})$ is the domain $\CD(\overline{{\mathfrak a}_\theta})$ of the form closure $\overline{{\mathfrak a}_\theta}$ of the quadratic form ${\mathfrak a}_\theta$ given by \eqref{ee} 
(see \cite[Theorems~VI.2.23, VI.2.1, and Corollary~VI.2.2]{MR0407617}). Hence for every $x\in \mathcal D(A_\theta^{1/2})$  
there exists a sequence $(x_n)_0^\infty\subset\mathcal D(A_\theta)$ such that $x_n\to x$ in $L^2(0,1)$, $n\to\infty$, and 
\begin{equation}
\label{tu}
  \big(A_\theta (x_m - x_n),x_m -x_n\big) = \langle x_m-x_n,x_m-x_n \rangle_\theta \longrightarrow 0, \quad m,n\to \infty.
\end{equation}
The relations \eqref{tu}, \eqref{inner} yield that $x_m'-x_n'\to 0$, $x_m/r-x_n/r\to 0$ in $L_2(0,1)$ and, if $c_\theta\ne 0$, also $x_m(1)-x_n(1)\to 0$ 
for $m,n\to\infty$. Since convergence  in $L_2(0,1)$ implies convergence almost everywhere, we can choose $r_0\in(0,1]$ such that $x_n(r_0)\to x(r_0)$,  $n\to\infty$; if $c_\theta\ne 0$, we can always choose $r_0=1$. Together with
$$
 x_n(r) - x_n(r_0) =\int_{r_0}^r\,x'(t)\d t, \quad r\in [0,1],
$$
it readily follows that the sequence $(x_n)_0^\infty$ converges uniformly in $[0,1]$ to $x$, that $x$ is absolutely continuous with $x' \in L_2(0,1)$, and that $x_n'\to x'$ in $L_2(0,1)$, $n\to \infty$. 
The relation \eqref{tu} also shows that the functions $x_n/r$ form a Cauchy sequence in $L_2(0,1)$ and hence also $x/r\in L_2(0,1)$. If $ \theta=\infty$, then  
$x_n \in \CD(A_\infty)$ implies that $x_n(1)=0$, $n\in \N$, and due to the uniform convergence of $(x_n)$ it follows that $x(1)=0$. Thus we have proved the inclusions ``$\subset$'' in \eqref{ob1} and \eqref{ob2}.

In order to prove the converse inclusions ``$\supset$'' in \eqref{ob1}, \eqref{ob2}, we equip the set $\CH_\theta$ on the right hand side of \eqref{ob1} or \eqref{ob2}, respectively, with the inner product $\langle\cdot,\cdot\rangle_\theta$ defined in~\eqref{inner}. Then $\CH_\theta$ becomes a Hilbert space and, by its definition, the subset $\CD(\ov{{\mathfrak a}_\theta})$ is a closed subspace of~$\CH_\theta$. Now assume that $y_0\in\CH_\theta$ is orthogonal to $\CD(\overline{{\mathfrak a}_\theta})$. Then, in particular, for every $x\in\CD_0(A_\theta)\subset\mathcal D(A_\theta)=\mathcal D({\mathfrak a}_\theta)$ with $\CD_0(A_\theta)$ given by \eqref{supp-comp},
\begin{equation}
\label{4}
  0=\langle x,y_0\rangle_\theta=(A_\theta x,y_0).
\end{equation}
Since $\mathcal D_0(A_\theta)$ is a core of $A_\theta$, we conclude that $y_0\in\mathcal D(A_\theta^*)=\mathcal D(A_\theta)$ and 
$A_\theta y_0 = 0$. The positivity of $A_\theta$ now implies $y_0=0$. This proves
that $\mathcal D\big(A_\theta^{1/2}\big)=\mathcal D(\overline{{\mathfrak a}_\theta})=\mathcal H_\theta$.

%and, by its definition, the subset $\CD(\overline{{\mathfrak a}_\theta})$ is a closed subset of $\CH_\theta$. Now assume that $y_0 \in \CH_\theta$
%is orthogonal to $\CD(\overline{{\mathfrak a}_\theta})$. Then, in particular, for every $x\in \CD_0(A_\theta) \subset \CD(A_\theta) = \CD({\mathfrak a}_\theta)$,
%\[
% 0 = \langle x, y_0 \rangle = x'(1) \ov{y_0(1)} + (A_\theta x, y_0). 
%\] 
%Consider first the cases $\theta=\infty$ and $\theta=0$. If $\theta=\infty$, then $y_0(1)=0$ since $y_0\in \CH_\infty$; if $\theta=0$, then $x'(1)=0$ since $x\in \CD(A_0)$. Hence, in both cases, $(A_\theta x, y_0)=0$ for every $x\in \CD_0(A_\theta)$. Since $\CD_0(A_\theta)$ is a core of $A_\theta$, we conclude that $y_0 \in \CD(A_\theta^\ast) = \CD(A_\theta)$ and $A_\theta y_0 = 0$. Because $A_\infty$ is positive, this implies $y_0=0$. This proves that $\CD(A_\theta^{1/2}) = \CD(\overline{{\mathfrak a}_\theta})=\CH_\theta$ for $\theta=\infty$ and $\theta=0$. For $\theta\in (0,\infty)$, we have
%\[
%  {\mathfrak a}_\theta = {\mathfrak a}_0 + {\mathfrak b}_\theta
%\] 
%where ${\mathfrak b}_\theta$ is the quadratic form given by ${\mathfrak b}_\theta[x]:= c_\theta |x(1)|^2$, $x\in\CD({\mathfrak b}_\theta) 
%:= \CD({\mathfrak a}_0)$ (compare \cite[p.~142]{MR2356213}). Since ${\mathfrak b}_\theta$ is relatively compact with respect to the form ${\mathfrak a}_\theta$ and hence relatively bounded with relative bound $<1$, it follows that $\CD(\ov{{\mathfrak a}_\theta})= \CD(\ov{{\mathfrak a}_0})=\CH_0 = \CH_\theta$ (see \cite[Examples~IV.1.8, IV.1.15]{MR0407617} and \cite[Theorem~VI.1.33]{MR0407617}).

ii) In \cite[Appendix II, Section 9.IV]{MR1255973} it was proved that the spectrum of $A_\infty$, and hence of every other operator $A_\theta$ (see \cite[Appendix~II, Theorem~6.2]{MR1255973} or \cite[Abschnitt~14.2]{MR2382320}), is discrete or, equivalently, $A_\theta$ has compact resolvent; moreover, all eigenvalues are simple. In fact, if we choose a fundamental system $\{\varphi(\cdot,\la),\psi(\cdot,\la)\}$ of $(\tau-\la)x=0$ such that
\[
  \begin{array}{rlrl}
   \varphi(1,\la)&\!\!\!\!=0, \quad &\psi(1,\la)&\!\!\!\!=1, \\[1mm]
   \varphi'(1,\la)&\!\!\!\!=-1, \quad & \psi'(1,\la)&\!\!\!\!=0,
  \end{array} 
\] 
where $'$ denotes the derivative with respect to the first variable $r$, then the corresponding Weyl-Titchmarsh function $m_\infty$ is given by
(see \cite[(4.8.2)]{MR0176151})
\[
  m_\infty(\la)= - \sqrt \la \ \frac{J_{l+1/2}'(\sqrt \la)}{J_{l+1/2}(\sqrt \la)} - {\frac 12} = 
  - \sqrt \la \ \frac{f_l'(1,\la)}{f_l(1,\la)}, \quad \la\in\C.\footnote{Note that, in the first edition of \cite{MR0176151} from 1946, the term $- 1/2$  was missing.}
\]
%\[
%  m_\infty(\la)= - \sqrt \la \ \frac{J_{l+1/2}'(\sqrt \la)}{J_{l+1/2}(\sqrt \la)}, \quad \la\in\C,
%\]
The Weyl-Titchmarsh function $m_\theta$ for $\theta \in [0,\infty)$ can be written as
\[
  m_\theta(\la) = \frac{1+ \theta \, m_\infty(\la)}{\theta - m_\infty(\la)}, \quad \la\in \C,
\]
(see \cite[Lemma~14.10]{MR2382320}).
Since the eigenvalues $\la_k(\theta)$ of $A_\theta$ are the poles of the Weyl-Titchmarsh function $m_\theta$, 
$\la_k(\infty)$ is the square of the $k$-th non-zero zero of the Bessel function $J_{l+1/2}$, while 
$\la_k(\theta)$ is the $k$-th non-zero zero of the function $m_\infty - \theta$ for $\theta \in [0,\infty)$. 
For all $\theta\in [0,\infty]$, the eigenspace of $A_\theta$ at an eigenvalue $\la_k(\theta)$ is spanned by the function
\[
 x_{k,\theta}(r) = f_l \bigl( r \sqrt{\la_k(\theta)} \bigr) 
 = \sqrt{\frac \pi 2} \sqrt{r \sqrt{\la_k(\theta)}} \, J_{l+1/2}\bigl( r \sqrt{\la_k(\theta)} \bigr), \quad r\in (0,1].
\]

iii) The function $m_\infty$ is a Nevanlinna function, that is, it is analytic in the upper half plane $\mathbb C^+$, maps $\mathbb C^+$ into $\mathbb C^+\cup \mathbb R$, and is complex symmetric with respect to $\R$ (see, e.g., \cite[Appendix II, Section 9.4 and Section~VI.59, Theorem 2]{MR1255973}); in particular, $m_\infty$ is real-valued on $\R$. Moreover, $m_\infty$ is meromorphic, its  singularities are the simple poles $\la_k(\infty)$, and it is strictly increasing between two neighbouring poles. This and the property that $\la_k(\theta)$ is the $k$-th zero of $m_\infty - \theta$ for $\theta \in [0,\infty)$ imply both claims of iii).

iv) The claimed inequality is an immediate consequence of Kre{\u\i}n's resolvent formula (see \cite{MR0018341}, \cite[Section~VIII.106]{MR1255973}, or \cite[Section~9]{gesztesy-09}).
\end{proof}

{
\renewcommand{\i}{{\rm i}}

\begin{remark}
By Proposition \ref{A_theta} i), for $\theta=\infty$ the boundary condition $x(1)=0$ carries over from $\CD(A_\infty)$ to $\CD(A_\infty^{1/2})$;  
it is called an \vspace{0.7mm} \emph{essential boundary condition} (comp.\ \cite[\S 7]{MR0024575}). For $\theta\in [0,\infty)$ the boundary condition $x'(1) + \theta \, x(1) = 0$ in $\CD(A_\theta)$ 
disappears in $\CD(A_\theta^{1/2})$; it is called a \emph{non-essential} or \emph{natural boundary condition} (see \cite[Section~10.5, p.~234]{MR1192782})). 
As a consequence, $\CD(A_\theta^{1/2})$ does not depend on $\theta$ for $\theta\in [0,\infty)$, whereas the action of $A_\theta^{1/2}$ does. 
\end{remark}

In order to study the problems \eqref{diffsyst}, \eqref{boundcond} and \eqref{diffsyst}, \eqref{boundcond-id}, 
we shall need, in particular, the eigenvalues of the operators $A_\infty$ and $A_l$: 

\begin{lemma}
\label{may}
The eigenvalues $\la_k(\theta)$, $\,k=1,2,\dots$ for $\theta=l$ and $\theta=\infty$ are as follows:
\begin{enumerate}
\item[{\rm i)}]  $\la_k(\infty)$ is the $k$-th non-zero zero of the function $\la \mapsto J_{l+1/2}(\sqrt{\la})$,
\item[{\rm ii)}] $\la_k(l)$ is the $k$-th non-zero zero of the function $\la \mapsto J_{l-1/2}(\sqrt{\la})$, 
\end{enumerate}
for $\,k=1,2,\dots$.
\end{lemma}

\begin{proof}
%Claims i) and ii) may either be deduced from the proof of Proposition~\ref{A_theta}~ii), noting that $\la_k(\theta)$ are the poles of the corresponding Weyl-Titchmarsh function $m_\theta$, or they may be %checked directly. 

%In fact, 
For $\theta=\infty$, the boundary condition $x_l(1)=0$ implies that $\la\in\C$ is an eigenvalue of $A_\infty$ if and only if $J_{l+1/2}(\sqrt{\la})=0$.
For $\theta=l$, the boundary condition $x_l'(1)+l x_l(1)=0$ implies that $\la\in\C$ is an eigenvalue of $A_l$ if and only if
\begin{equation}
\label{Bessel-diff}
 0 =  \frac{\d}{\d r} \big( \sqrt{r} \,J_{l+1/2}(r\sqrt{\la}) \big) + l \, \sqrt r \,J_{l+1/2}(r\sqrt{\la}) \Big|_{r=1}=  \sqrt{\la} \, J_{l-1/2}(\sqrt{\la});
\end{equation}
here we have used the differentiation formula $\big( z^l \!\cdot\! z\, j_l(z) \big)' = z^{l+1} j_{l-1}(z)$ for the spherical Bessel functions $j_l$ (see \cite[10.1.23]{abramowitz+stegun}). 
\end{proof}

\begin{remark}
\label{may-add}
Note that the index $\theta$ in $A_\theta$ and its eigenvalues $\la_k(\theta)$ only refers to the boundary condition $x'(1)+\theta x(1)=0$. 
Since the differential expression $\tau=-\partial_r^2+l(l+1)/r^2$ defining $A_\theta$ always depends on $l$, so do all eigenvalues $\la_k(\theta)$ (compare Lemma \ref{may}). As $l\in\N$ is assumed to be fixed throughout the paper, we do not indicate this dependence in general.
%The only exception is the explicit comparison of modes with different $l_1\neq l_2$ contained in the anti-dynamo theorem (6.3).
\end{remark}

For the differential expression $\tau_\alpha$ in \eqref{taus}, we use the relation 
\[
  \tau_\alpha = \alpha \,\tau - \alpha' \partial_r, 
\]
thus relating the operator $A_{\theta,\alpha}$ induced by $\tau_{\!\alpha}$ to the operator $A_\theta$ induced by~$\tau$. 

\begin{lemma}
\label{march3}
Let $\theta\in [0,\infty]$ and let $D$ be the operator of differentiation in $L_2(0,1)$, 
\[
  \CD(D) := W_2^1(0,1), \quad Dx := x'.
\]
Then the operators $D A_\theta^{-1/2}$ and $D A_\theta^{-1}$ are  defined on $L_2(0,1)$ and bounded~with
\begin{equation}\label{dates} 
  \| D A_\theta^{-1/2} \| \le 1, \quad \bigl\|D A_\theta^{-1} \bigr\| \le  \frac 1{\sqrt{\la_1(\theta)}}.
\end{equation}
\end{lemma}

\begin{proof}
Proposition \ref{A_theta} i) and its proof yield that $\CD(A_\theta^{1/2}) = \CD\big(\overline{{\mathfrak a}_\theta}\big) \subset \CD(D)$~and
\begin{align*}
  \| A_\theta^{1/2} x \|^2 = \overline{{\mathfrak a}_\theta}[x] \ge \int_0^1 \bigl| x'(r) \bigr|^2 \d r = \| D x\|^2, \quad x\in\CD(A_\theta^{1/2}), 
\end{align*}
that is, $\big\|D A_\theta^{-1/2}\big\| \le 1$. The claims for $D A_\theta^{-1}$ now follow from the identity 
$D A_\theta^{-1}=D A_\theta^{-1/2} A_\theta^{-1/2}$ and from the estimate $\|A_\theta^{-1/2}\| \le 1 /\sqrt{\la_1(\theta)}$.
\end{proof}

\begin{proposition}
\label{A_theta-alpha}
Let $\theta\in[0,\infty]$. The linear operator $A_{\theta,\alpha}$ in Definition {\rm \ref{A_theta-alpha-def} }
is densely defined, symmetric and hence closable, and it satisfies 
\begin{equation}
\label{opid}
  A_{\theta,\alpha} = \alpha A_\theta - \alpha' D.
\end{equation}
\end{proposition}

\begin{proof}
First we have to show that $A_{\theta,\alpha}$ is well-defined, that is, $\alpha x' \in {\rm AC_{loc}}((0,1])$ and $\tau_\alpha x \in L_2(0,1)$ 
for $x\in \CD(A_\theta)$. To this end, let $x\in \CD(A_\theta)$. By definition, this implies $x'\in {\rm AC_{loc}}((0,1])$ and $\tau x \in L_2(0,1)$.
Since $\alpha \in C^1([0,1])$, it follows that $\alpha x' \in {\rm AC_{loc}}((0,1])$.
Further, Lemma \ref{march3} shows that $\CD(A_\theta) \subset \CD(D)$ and hence $x'\in L_2(0,1)$.
Together with $\alpha, \alpha' \in C([0,1])$, we obtain
\[
 \tau_\alpha x = \alpha \tau x - \alpha' x' \in L_2(0,1),
\]
and also the operator identity \eqref{opid}. For the symmetry of $A_{\theta,\alpha}$, it suffices to show that $(A_{\theta,\alpha}x,x)\in \R$ for all $x\in \CD(A_{\theta,\alpha})=\CD(A_\theta)$. Using the boundary condition at $r=1$, we see that
\begin{align*}
 (A_{\theta,\alpha} x, x) = & \ \alpha(0) \, \lim_{\varepsilon \to 0} \big( x'(\eps) \,\overline{x(\eps)} \big) + c_\theta \, \alpha(1) \,|x(1)|^2 \\
                            & \ + \int_0^1 \!\alpha(r) \bigl| x'(r) \bigr|^2 \d r + \int_0^1 \!\alpha(r) \frac{l(l+1)}{r^2} |x(r)|^2 \d r, 
\end{align*}
where $c_\theta =\theta$ for $\theta \in [0,\infty)$ and $c_\infty=0$. Since $x\in \CD(A_\theta)$ and since the differential expression $\tau$ defining $A_\theta$ is in limit point case at $0$, the limit in the first term on the right hand side is real (see \cite[Satz~13.19]{MR2382320}). 
Because $\alpha$ is real-valued, it follows that $(A_{\theta,\alpha}x,x)\in \R$. 
\end{proof}

We close this section with some resolvent estimates for the self-adjoint operators~$A_\theta$, which will be used in the next sections.

\begin{lemma}
\label{resolv-est}
For $\theta \in [0,\infty]$ and $\la\notin\big(\!-\!\infty,-\la_1(\theta)\big]$, we have the norm estimates 
\begin{align*}
  \big\| (A_\theta +\la)^{-1} \big\| & \ \
  \left\{ \begin{array}{lll} \ds = \frac 1 {|\la_1(\theta)+\la|} & \text{if \ }  \Re \la \ge - \la_1(\theta), \qquad & {\rm (i)}\\[3mm]
                             \ds \le \frac 1{|\Im \la|} & \text{if \ }  \Re \la \le - \la_1(\theta), \qquad & {\rm (ii)}
          \end{array}
  \right.\\[1mm]
  \big\| A_\theta^{1/2} (A_\theta +\la)^{-1} \big\| & \ \
  \left\{ \begin{array}{lll} \ds = \frac {\sqrt{\la_1(\theta)}} {|\la_1(\theta)+\la|} & \text{if \ }  |\la| \le \la_1(\theta), \qquad & {\rm (a)}\\[3mm]
                             \ds \le \frac {\sqrt{|\la|}} {||\la|+\la|}  & \text{if \ }  |\la| \ge \la_1(\theta), \qquad & {\rm (b)}
          \end{array}
  \right.\\[1mm] 
  %\intertext{}
  \big\| A_\theta (A_\theta +\la)^{-1} \big\| & \ \
  \left\{ \begin{array}{lll} = 1 & \text{if \ }  \Re \la \ge 0 , \ \ & {\rm (1)}\\[2mm]
                             \ds \le \frac {|\la|}{|\Im \la|} & \text{if \ }  \Re \la \le 0, \, |\la+ \la_1(\theta)/2| \ge \la_1(\theta)/2,  & {\rm (2)}\\[3mm]
                             \ds = \frac {\la_1(\theta)} {|\la_1(\theta)+\la|} & \text{if \ }   |\la+ \la_1(\theta)/2| \le \la_1(\theta)/2; \ \ & {\rm (3)}
          \end{array}
  \right.
\end{align*}
in each of the three cases, the bounds on the right-hand side define a continuous function of $\lambda$ on $\C\setminus\big(\!-\!\infty,-\la_1(\theta)\big]$.
\end{lemma}

\begin{proof}
The resolvent of the self-adjoint operator $A_\theta$ satisfies the well-known relation (see, e.g., \cite[Section~V.3.5]{MR0407617})
\begin{equation}
\label{kato}
  \big\| A_\theta^s(A_\theta+\la)^{-1 }\big\| = \sup_{t\in \sigma(A_\theta)} \frac{|t|^s}{|t+\la|}, \quad s\in [0,1], \ \ \la \notin \sigma(A_\theta).
\end{equation}
For $s=0$, the right hand side of \eqref{kato} equals $1 / \dist (-\la, \sigma(A_\theta))$ and, together with the inclusion 
$\sigma(A_\theta) \subset [\la_1(\theta),\infty)$, the equality (i) and the estimate (ii) follow.

For $s=1/2$, it is not difficult to see that the function $f(t):= \sqrt{t}/|t+\la|$, $t \in [0,\infty)$, has a local maximum at $t_\la=|\la|$. Hence, if
$|\la|\le \la_1(\theta)$, then the function $f$ restricted to $\sigma(A_\theta) \subset [\la_1(\theta),\infty)$ attains its maximum at $t=\la_1(\theta) \in \sigma(A_\theta)$ and the equality (a) follows; if $|\la| > \la_1(\theta)$, then $f$ restricted to $\sigma(A_\theta) \subset [\la_1(\theta),\infty)$ is estimated by its maximum on $[\la_1(\theta),\infty)$ attained at $t=|\la|$ and (b) follows.  

For $s=1$, a short calculation yields that the function $g(t):= t/|t+\la|$, $t \in [0,\infty)$, is monotonically increasing if $\Re\la \ge 0$ and has a local maximum at $t_\la=|\la|^2/|\Re\la|$ if $\Re\la<0$. Thus, if $\Re\la \ge 0$, then the function $g$ restricted to $\sigma(A_\theta) \subset [\la_1(\theta),\infty)$ attains its maximum at $\infty$ and (1) follows. If $\Re\la<0$, we note that 
\[
  t_\la= \frac{|\la|^2}{|\Re\la|} \le \la_1(\theta) \iff |\la|^2 \!+ \Re\la \cdot\la_1(\theta) \le 0 \iff |\la+ \la_1(\theta)/2| \le \frac{\la_1(\theta)}2,
\]
and therefore in this case the function $g$ restricted to $\sigma(A_\theta) \subset [\la_1(\theta),\infty)$ attains its maximum at $t=\la_1(\theta) 
\in \sigma(A_\theta)$ and the equality (3) follows; if $|\la+ \la_1(\theta)/2| > \la_1(\theta)/2$, then $g$ restricted to $\sigma(A_\theta) \subset [\la_1(\theta),\infty)$ is estimated by its maximum on $[\la_1(\theta),\infty)$ attained at $t_\la=|\la|^2/|\Re\la|$. Now another short calculation shows that 
$g(t_\la)$ coincides with the upper bound in (2).  
\end{proof}

\section{The spectrum of $\CA_\theta$ and a first eigenvalue estimate}
\label{section3}

In the following we present a first perturbational approach to study the spectral properties of the operator $\CA_\theta$ associated with the dynamo problem. To this end, we regard the block operator matrix $\CA_\theta$ as a bounded perturbation of its lower triangular part. We use this decomposition to show that $\CA_\theta$ is a closed operator with compact resolvent and to establish estimates for its eigenvalues.

\begin{theorem}
\label{harz}
Let $\theta\in[0,\infty]$ and define
\begin{alignat}{2}
\label{decomp2a}
  \CQ_\theta & := \matrix{cc}{-A_\theta & 0 \\ \hspace{4.5mm} A_{\theta,\al} & - A_\infty}, 
  & \quad \CD(\CQ_\theta) & = \CD(A_\theta)\oplus\CD(A_\infty),\\
\label{decomp2b}  
  \quad \CR & := \matrix{cc}{0 & \al \\ 0 & 0 }, 
  & \quad \CD(\CR) &= L_2(0,1) \oplus L_2(0,1).
\end{alignat}   
Then the block operator matrix 
\begin{equation} 
\label{decomp2}
  \CA_\theta=\CQ_\theta + \CR
\end{equation} 
is closed and has compact resolvent; its spectrum is symmetric to $\R$ and consists of isolated eigenvalues of finite algebraic multiplicities with no finite accumulation~point.
\end{theorem}

\begin{proof}
Since $\alpha\in C([0,1])$, the operator $\CR$ is everywhere defined and bounded. Moreover, $\CD(\CQ_\theta)= \CD(\CA_\theta)$ and hence the operator identity \eqref{decomp2} holds. 

Next we show that 
\begin{equation}
\label{sig0}
  \rho(\CQ_\theta) = \rho(-A_\theta) \cap \rho(-A_\infty) \ne \emptyset, 
\end{equation}
which implies the closedness of $\CQ_\theta$, and that $(\CQ_\theta - \la)^{-1}$ is compact for all $\la\in\rho(Q_\theta)$.
%$\CQ_\theta$ is closed and has compact resolvent.
%By Lemma \ref{march3}, the operator $DA_\theta^{-1}$ and hence $D (A_\theta+\la)^{-1}=  \la^{-1} DA_\theta^{-1} (A_\theta^{-1} + \la^{-1})^{-1}$ is bounded for
%$\la \in \rho(-A_\theta)$, $\la \ne 0$. Since $\al$, $\al' \in C([0,1])$, relation \eqref{opid} thus shows that
By \eqref{opid}, we~have
\begin{equation}
\label{combi}
\begin{array}{rl}
  A_{\theta,\al}(A_\theta+\la)^{-1} \!\!\!\!& = \Big( \al (A_\theta +\la) - \al \la - \al' D \Big) (A_\theta+\la)^{-1}\\
                           & =  \al - \big( \al \la + \al' D \big) (A_\theta+\la)^{-1}.
\end{array}
\end{equation}
Since $\al$, $\al' \in C([0,1])$ and $D (A_\theta+\la)^{-1} $ is bounded for
$\la \in \rho(-A_\theta)$ by Lemma~\ref{march3}, the operator $A_{\theta,\al}(A_\theta+\la)^{-1}$ is bounded 
for every $\la\in \rho(-A_\theta)$. Therefore, for $\la\in \rho(-A_\theta)\cap \rho(-A_\infty)$, the inverse
\begin{equation}
\label{may10}
  \big( \CQ_\theta - \la \big)^{-1} = \matrix{cc}{ -(A_\theta+\la)^{-1} & 0 \\[1mm] -(A_\infty+\la)^{-1} A_{\theta,\al} (A_\theta+\la)^{-1} & -(A_\infty+\la)^{-1}}
\end{equation}
is bounded and everywhere defined. This proves the inclusion ``$\supset$'' in \eqref{sig0}. Vice versa, assume that $\la \!\in\! \sigma(-A_\theta) \cup \sigma(-A_\infty) =  
\sigma_{\rm p}(-A_\theta) \cup \sigma_{\rm p}(-A_\infty)$. If $\la\!\in\! \sigma_{\rm p}(-A_\infty)$, then there is an element 
$y_2 \in \ker (-A_\infty - \la) \setminus\{0\}$. Then $ (0, y_2)^{\rm t} \in \ker (\CQ_\theta - \la) \setminus\{0\}$ and hence $\la\in \sigma_{\rm p}(\CQ_\theta) \subset \sigma(\CQ_\theta)$.
If $\theta \in [0,\infty)$ and $\la \in \sigma_{\rm p}(-A_\theta)$, then there is an element $y_1 \in \ker (-A_\theta - \la) \setminus\{0\}$. Since $\theta \not= \infty$, we have
$\la \notin \sigma(-A_\infty)$ by Proposition \ref{A_theta} iii) and so $(y_1, (A_\infty +\la)^{-1} A_{\theta,\al} y_1 )^{\rm t} \in \ker (\CQ_\theta - \la) \setminus\{0\}$, that is, $\la\in \sigma_{\rm p}(\CQ_\theta) \subset \sigma(\CQ_\theta)$. This proves \eqref{sig0}.

By Proposition \ref{A_theta} ii) the inverses $(A_\theta+\la)^{-1}$ are compact 
and we have shown that $A_{\theta,\al}(A_\theta+\la)^{-1}$ is bounded for all $\la \in \rho(-A_\theta)$ and $\theta \in [0,\infty]$.
Thus, by equation \eqref{may10}, the inverses $(\CQ_\theta-\la)^{-1}$ are compact as well. 

Since $\CA_\theta$ is a bounded perturbation of $\CQ_\theta$, it is immediate that $\CA_\theta$ is closed.~In order to show that $\CA_\theta$ has compact resolvent, it suffices to show that $(\CA_\theta - \la)^{-1}$ exists and is compact for some $\la \in \C$ (see, e.g., \cite[Theorems~IV.1.1 and III.6.29]{MR0407617}). For $\la \ge 0$, we can write 
\begin{equation}
\label{db}
   \CA_\theta - \la = \big( I + \CR \big( \CQ_\theta - \la\big)^{-1} \big) \big( \CQ_\theta - \la \big).
\end{equation} 
By Lemma \ref{march3} and by Lemma \ref{kato} (i), (1), we have the uniform estimates
\[
\begin{array}{rl}
  \big\| \la (A_\theta+\la)^{-1} \big\| \!\!\!\!&=  \ds \frac{\la}{\la+\la_1(\theta)} \le 1, \\
  \big\| D (A_\theta+\la)^{-1} \big\| \!\!\!\!&= \big\| DA_\theta^{-1} \big\| \big\| A_\theta (A_\theta+\la)^{-1} \big\| \le \ds \frac 1 {\sqrt{\la_1(\theta)}}, 
\end{array}
\quad \la \ge 0.
\]
Together with  \eqref{combi}, it follows that $A_{\theta,\al} (A_\theta+\la)^{-1}$ is uniformly bounded for $\la\ge 0$.
Thus, \eqref{may10} and the fact that $\|(A_\theta + \la)^{-1} \| \to 0$ for $\la \to \infty$ and $\theta\in[0,\infty]$ show that 
$\| ( \CQ_\theta - \la )^{-1} \| \to 0$ for $\la\to\infty$ and $\theta\in[0,\infty]$. 
As a consequence, we can choose $\la_0 > 0$ sufficiently large such that 
$\|\CR ( \CQ_\theta - \la )^{-1} \| <1$ for $\la \ge \la_0$ and~so
\[
  \big( \CA_\theta - \la \big)^{-1} = \big( \CQ_\theta + \CR  - \la \big)^{-1} = \big( \CQ_\theta - \la \big)^{-1} \big( I + \CR \big( \CQ_\theta - \la\big)^{-1} \big)^{-1}
\]
is compact for $\la \ge \la_0$.

The symmetry of the (point) spectrum of $\CA_\theta$ is evident since the entries of the block operator matrix $\CA_\theta$ are differential operators with real coefficients.
\end{proof}

Although $\CA_\theta$ is a bounded perturbation of $\CQ_\theta$, the norm of the perturbation being $\|\al\|$, 
we cannot conclude that $\sigma(\CA_\theta)$ lies in a $\|\alpha\|$-neighbourhood of 
$\sigma(\CQ_\theta) = \sigma(-A_\theta) \cup \sigma(-A_\infty)$ 
since $\CQ_\theta$ is neither self-adjoint nor normal. Therefore, in the next proposition, we 
use a Neumann series argument to exclude points $\la\in\C$ from the (point) spectrum of $\CA_\theta$.

\begin{proposition}
\label{thm1}
Let $\theta\in[0,\infty]$. If $\,\la\notin \sigma(-A_\theta) \cup \sigma(-A_\infty)$ and 
\begin{equation}
\label{neumann-new}
  \big\|\alpha (A_\infty\!+\lambda)^{-1}(\al A_\theta -\alpha'D)(A_\theta+\lambda)^{-1} \big\| < 1,
\end{equation}
then $\la \in \rho(\CA_\theta)$.
\end{proposition}

\begin{proof}
By \eqref{db}, we have $\lambda\in\rho(\mathcal A_\theta)$ if and only if $\lambda\in\rho(\CQ_\theta)=\rho(- A_\theta)\cap\rho(- A_\infty)$ 
and $I+\CR(\CQ_\theta - \lambda)^{-1}$ is boundedly invertible. 
By \eqref{may10}, the operator
\begin{align*}
I+\CR(\CQ_\theta - \lambda)^{-1}&=I+\begin{pmatrix}0&\alpha\\0&0\end{pmatrix}\begin{pmatrix}-(A_\theta+\lambda)^{-1}&0\\-(A_\infty+\lambda)^{-1}A_{\theta,\alpha}(A_\theta+\lambda)^{-1}&-(A_\infty+\lambda)^{-1}\end{pmatrix}\\
&=\begin{pmatrix}I-\alpha (A_\infty+\lambda)^{-1}A_{\theta,\alpha}(A_\theta+\lambda)^{-1}&-\alpha(A_\infty+\lambda)^{-1}\\0&I\end{pmatrix}
\end{align*}
is boundedly invertible if and only if the operator $I-\alpha (A_\infty+\lambda)^{-1}A_{\theta,\alpha}(A_\theta+\lambda)^{-1}$ is boundedly invertible. The latter 
holds if condition \eqref{neumann-new} is satisfied.
\end{proof}

\begin{remark}
The block operator matrix $\CA_\theta$ can also be decomposed as
\[
   \CA_\theta := \matrix{cc}{ -A_\theta & 0 \\  0 & - A_\infty}  + \matrix{cc}{0 & \al \\ A_{\theta,\alpha} & 0}.
\]
Here the first term is self-adjoint, while the second term is only relatively bounded with respect to the first one. 
A corresponding Neumann series argument leads to the same condition \eqref{neumann-new} for a point to be in the resolvent set. 
%The compactness of the resolvent of $\CA_{\theta}$ is easier to see with the decomposition $\CA_\theta=\CQ_\theta+\CR$ used~above.
\end{remark}

In the following we use Proposition \ref{thm1} to obtain an enclosure for the eigenvalues of $\CA_\theta$. To this end, we
estimate the norm on the left hand side of \eqref{neumann-new} by means of Lemma \ref{resolv-est}. 
According to the different resolvent estimates therein, we decompose the complex plane $\C$ as
\begin{equation}
\label{z-regions}
  \C = Z_0 \cup Z_1 \cup Z_2 \cup Z_3 \cup Z_4 \cup Z_5 \cup Z_6
\end{equation}
into the pairwise disjoint sets 
\begin{align*}
  Z_0 := \biggl\{ z\in\C : &\ \Re z \le -\la_1(\theta), \ \Im z =0 \biggr\} = \big( -\infty, -\la_1(\theta) \big], \\
  Z_1 := \biggl\{ z\in\C : &\ \Re z \le -\la_1(\infty), \ \Im z \ne 0 \biggr\}, \\
  Z_2 := \biggl\{ z\in\C : &\ -\la_1(\infty) < \Re z \le 0, \ \Im z \ne 0, \ |z| > \la_1(\theta) \biggr\}, \\
  Z_3 := \biggl\{ z\in\C : &\ -\la_1(\theta) < \Re z \le 0, \ |z + \la_1(\theta)/2| > \la_1(\theta)/2, \ |z| \le \la_1(\theta) \biggr\}, \\
%\intertext{}  
  Z_4 := \biggl\{ z\in\C : &\ -\la_1(\theta) < \Re z \le 0, \ |z + \la_1(\theta)/2| \le \la_1(\theta)/2 \biggr\}, \\
  Z_5 := \biggl\{ z\in\C : &\ \Re z > 0, \ |z| \le \la_1(\theta) \biggr\}, \\
  Z_6 := \biggl\{ z\in\C : &\ \Re z > 0, \ |z| > \la_1(\theta) \biggr\}.
\end{align*} 
The sets $Z_i$ are shown in Figure \ref{Z} for $\theta\!=\!l\!=\!1$ where $\la_1(\theta) \!= \!\pi^2 \!\approx\! 9.87\!$, $\la_1(\infty)\!\approx\! 20.19$.

%Note that, for $\theta<\infty$, we have $\la_1(\theta) < \la_1(\infty)$ by Proposition \ref{A_theta} iii) and hence these sets 
%are pairwise disjoint; for $\theta=\infty$, the only common point of two of these sets is $-\la_1(\infty) \in  Z_5 \cap Z_6$.

\begin{figure}[h]
\vspace{1mm}
\parbox{6cm}
{\epsfysize=6cm
\epsfbox{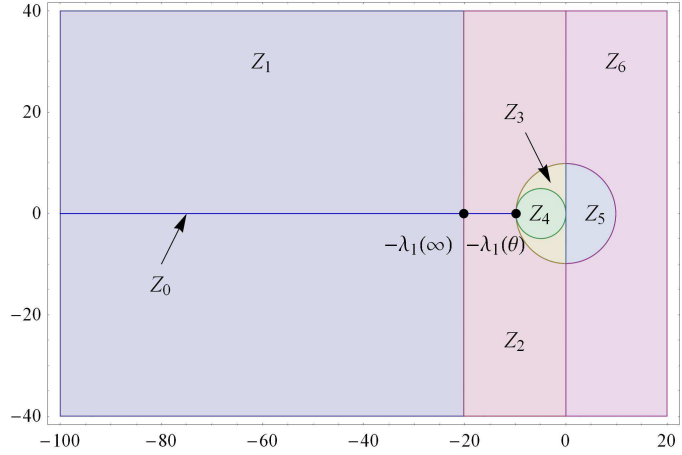}}
\caption{\label{Z} Decomposition \eqref{z-regions} of the complex plane for $l=1$.}
\end{figure} 

\begin{proposition}
\label{thm2}
Let $\theta\in[0,\infty]$. Then %$\la \notin \sigma(\CA_\theta)=\sigma_{\rm p}(\CA_\theta)$ if 
\begin{equation}
\label{Sigma}
 \sigma(\CA_\theta) \subset \Sigma_\theta :=  \big(\!-\!\infty,-\la_1(\theta)\big] \cup \big\{ \la \in \C: f(\la) \ge 1 \big\}
 %\Sigma_0 \cup \Sigma_1 \cup \Sigma_2 \cup \Sigma_3 \cup \Sigma_4 \cup \Sigma_5 \cup \Sigma_6
\end{equation}
where the function $f:\C\setminus (-\infty,-\la_1(\theta)] \to [0,\infty)$ is defined as 
\begin{equation}
\label{f}
  f(\la) := f_j(\la), \quad \la \in Z_j, \ j=1,2,\dots,6,
\end{equation}
%\begin{align*}
% \Lambda_0 &:= Z_0=\big(\!-\!\infty, -\la_1(\theta) \big], \qquad
% \Lambda_j := \big\{ \la \in Z_j : f_j(\la) \ge 1 \big\}, \quad j=1,2,\dots,6,
%\end{align*} 
with
%and the functions $f_j:Z_j \to [0,\infty)$ are given by 
\begin{alignat*}{2} 
 f_1(\la) &:=
 \left( \|\al\|^2 \frac{|\la|}{|\Im \la|}+ \|\al\| \, \|\al'\| \frac {\sqrt{|\la|\,}}{|\la + |\la|\,|} \right) \frac 1{|\Im\la|}, 
 && \quad \la \in Z_1, \\
 f_2(\la) &:=
 \left( \|\al\|^2 \frac{|\la|}{|\Im \la|}+ \|\al\| \, \|\al'\| \frac {\sqrt{|\la|\,}}{|\la + |\la|\,|} \right) \frac 1{|\la+\la_1(\infty)|}, 
 && \quad \la\in Z_2,\\  
\intertext{ }\\[-1.2cm]
 f_3(\la) &:=
 \left( \|\al\|^2 \frac{|\la|}{|\Im \la|}+ \|\al\| \, \|\al'\| \frac {\sqrt{\la_1(\theta)}}{|\la + \la_1(\theta)|} \right) \frac 1{|\la+\la_1(\infty)|}, 
 && \quad \la \in Z_3,\\
   f_4(\la) &:=
 \left( \|\al\|^2 \frac {\la_1(\theta)}{|\la + \la_1(\theta)|} + \|\al\| \, \|\al'\|  \frac {\sqrt{\la_1(\theta)}}{|\la + \la_1(\theta)|} \right) \frac 1{|\la+\la_1(\infty)|}, 
 && \quad \la \in Z_4,\\ 
 f_5(\la) &:=
 \left( \|\al\|^2 + \|\al\| \, \|\al'\| \frac {\sqrt{\la_1(\theta)}}{|\la + \la_1(\theta)|} \right) \frac 1{|\la+\la_1(\infty)|},
 && \quad \la \in Z_5,\\
  f_6(\la) &:=
 \left( \|\al\|^2 + \|\al\| \, \|\al'\| \frac {\sqrt{|\la|\,}}{|\la + |\la|\,|} \right) \frac 1{|\la+\la_1(\infty)|},
 && \quad \la \in Z_6.
\end{alignat*}
The set $\Sigma_\theta$ is symmetric to $\R$ and bounded from the right with 
\begin{equation}
\label{atheta}
    a_\theta := \max \Re \Sigma_\theta \ge - \la_1(\theta),
\end{equation}
but it is not bounded from above nor from below in the complex plane $\C$.
\end{proposition}

\begin{proof}
By Lemma \ref{march3}, we have $\big\| D A_\theta^{-1/2} \big\|\le 1$. Hence 
condition \eqref{neumann-new} in Proposition~\ref{thm1} is clearly satisfied for $\la\notin (-\infty,-\la_1(\theta)]$  if
\begin{equation}
\label{important}
  \big\| (A_\infty\!+\!\lambda)^{-1} \big\| \left( \|\al\|^2 \big\| A_\theta (A_\theta\!+\!\lambda)^{-1} \big\| 
  + \| \al\| \, \|\al'\| \, \big\| A_\theta^{1/2}(A_\theta\!+\!\lambda)^{-1} \big\| \right) < 1.
\end{equation}
Now we obtain the spectral inclusion \eqref{Sigma} from Proposition~\ref{thm1} 
by combining the various resolvent estimates in Lemma~\ref{resolv-est}, e.g., the estimates (ii), (b), and (2) for~$Z_1$.

The symmetry of $\Sigma_\theta$ follows since $\la_1(\theta)$, $\la_1(\infty) \in \R$ and hence $f_j(\la)=f_j(\ov\la)$, $\la\notin\C\setminus(-\infty,-\la_1(\theta)]$. 
Because $f_6(\la) \to 0$ for $\Re \la \to \infty$, the set $\Sigma_\theta$ is bounded from the right. 
The inequality \eqref{atheta} is immediate from the inclusion $\big(\!-\!\infty, -\la_1(\theta) \big] \subset \Sigma_\theta$.
If $\Sigma_\theta$ were bounded from above or from below, there would exist an $M_0 \ge 0$ such that $f_1(\la) < 1 $ for 
all $\la \in \C$ with $|\Im\la| > M_0$ and $\Re\la \le -\la_1(\infty)$; on the other hand, 
\[
 f_1(\la) \ge \|\alpha\|^2 \frac{|\la|}{|\Im \la|^2} \ge \|\alpha\|^2 \frac{|\Re\la|}{M_0^2} \longrightarrow \infty, \quad \Re \la \to - \infty,
\]
a contradiction.
\end{proof}
 
%\begin{remark}
%On $\C\setminus (-\infty,-\la_1(\theta))$, every function $f_j$ can be extended by continuity to the boundary $\partial Z_j$ of $Z_j$. Denoting this extension is again by $f_j$, we have 
%\[
%f_i(\la) = f_j(\la), \quad \la\in  \big( \partial Z_i \cap \partial Z_j \big) \setminus (-\infty,-\la_1(\theta)),
%\] 
%by the last claim of Lemma \ref{resolv-est}.
%\end{remark} 
% 
In the sequel we describe the shape and the location of the set $\Sigma_\theta$ in dependence of $\|\alpha\|$ and $\|\alpha'\|$ for $\alpha\not\equiv 0$;
if $\alpha \equiv 0$, then clearly $\Sigma_\theta=\big(\!-\!\infty, -\la_1(\theta) \big]$.

\begin{lemma}
\label{implicit}
Let $f$ be given by \eqref{f}. Then
$\varphi:\R^2 \setminus \big( (-\infty,-\la_1(\theta)] \times \{0\} \big) \to [0,\infty)$ defined by
\[
  \varphi(\xi,\eta) := f(\la), \quad \xi+\i\eta = \la  \in \C\setminus (-\infty,-\la_1(\theta)],
\]
is continuous and continuously differentiable on $\R^2 \setminus \big( (-\infty,-\la_1(\theta)] \times \{0\} \big)$ with
\begin{equation}
\label{curves}
  \dfrac{\partial \varphi}{\partial \xi}(\xi,\eta) < 0, \quad \dfrac{\partial \varphi}{\partial \eta}(\xi,\eta) < 0, \quad  \xi \in \R, \ \eta > 0,  
\end{equation}
and
\begin{equation}
\label{curves2}
  \dfrac{\partial \varphi}{\partial \xi}(\xi,0) < 0, \quad \dfrac{\partial \varphi}{\partial \eta}(\xi,0) =0, \quad  \xi \in \R \setminus (-\infty,-\la_1(\theta)].
\end{equation}
\end{lemma}

\begin{proof}
According to the definition of $f$ in \eqref{f}, the function $\varphi$ is given by
\[
 \varphi(\xi,\eta) = \varphi_j(\xi,\eta):=f_j(\la), \quad \la=\xi+\i \eta \in Z_j, \quad j=1,2,\dots,6.
\]
Hence $\varphi$ is continuous since the functions $\varphi_j$ are continuous on their domains $Z_j$ and, by the last claim in Lemma \ref{resolv-est}, the continuous extensions of two functions $\varphi_i$ and $\varphi_j$ coincide on common boundary points of $Z_i$ and $Z_j$.

It is elementary, but very tedious, to calculate the partial derivatives of the functions $\varphi_j$ on $Z_j$
and to show that the first partial derivatives of $\varphi_i$ and $\varphi_j$ coincide on common boundary points of $Z_i$ and $Z_j$. This proves that
$\varphi$ has continuous first partial derivatives on $\R^2 \setminus \big((-\infty,-\la_1(\theta)]\times \{0\} \big)$ and is hence totally differentiable there.
These calculations also show that \eqref{curves} and \eqref{curves2} hold. 
\end{proof}

In the next theorem the boundary of the set $\Sigma_\theta$ is described by a function $h_\theta$; in particular, 
we derive formulas for the rightmost point $a_\theta$ of $\Sigma_\theta$ (see \eqref{atheta}).

\begin{theorem}
\label{!!!}
Let $\theta\in[0,\infty]$ and $\al \not \equiv 0$.
Then there exists a continuous strictly decreasing function $h_\theta\!:(\!-\infty,a_\theta] \to [0,\infty)$ with
\begin{equation}
\label{htheta}
  \lim_{t\to-\infty} h_\theta (t) = \infty, \quad h_\theta(a_\theta)=0, \quad \lim_{t\nearrow a_\theta} h'_\theta (a_\theta)=-\infty
\end{equation}
such that
\begin{equation}
\label{Sigma-new}
 \sigma(\CA_\theta) \subset \Sigma_\theta = \big\{ \la \in \C : \, \Re \la \le a_\theta, \, |\Im\la| \le h_\theta (\Re \la) \big\}.
\end{equation}
Depending on $\|\alpha\|$ and $\|\alpha'\|$, the following cases occur: \vspace{1mm}
\begin{hlist}
\item[{\rm (i)}] $\,0< \|\alpha\|^2 + \dfrac{\|\al\|\,\|\alpha'\|}{\sqrt{\lambda_1(\theta)}} \le \lambda_1(\infty)${$:$} \ \vspace{-1mm} Then 
$$
  -\la_1(\theta) <  a_\theta \le 0
$$ 
and $a_\theta$ is given by
\[
\hspace*{6mm} 
  a_\theta  = - \frac{\la_1(\infty) \!+\! \la_1(\theta)}2 
  + \sqrt{\left( \frac{\la_1(\infty) \!-\! \la_1(\theta)}2 \right)^{\!2} \!\!\!+ \la_1(\theta) \!\left( \|\alpha\|^2 \!+    \dfrac{\|\al\|\|\alpha'\|}{\sqrt{\lambda_1(\theta)}}\right)}.
\]
\item[{\rm (ii)}] $\,\|\alpha\|^2 \!+ \dfrac{\|\al\|\,\|\alpha'\|}{\sqrt{\lambda_1(\theta)}}>\lambda_1(\infty)\,$ and $\,\|\alpha\|^2 \!+ \dfrac{\|\al\|\,\|\alpha'\|}{2\sqrt{\lambda_1(\theta)}} \le \lambda_1(\infty)+ \la_1(\theta)${$:$} \ Then
$$
  0 <  a_\theta  \le \la_1(\theta)
$$
and $a_\theta$ is given by
\[
\hspace*{6mm} 
a_\theta = - \frac{\la_1(\infty) \!+\! \la_1(\theta)\!-\!\|\al\|^2}2 
  + \sqrt{\left( \frac{\la_1(\infty) \!-\! \la_1(\theta)-\|\al\|^2}2 \right)^{\!2} \!\!\!+ \|\al\|\|\alpha'\| \sqrt{\lambda_1(\theta)}};
\]
\noindent
moreover, $\Sigma_\theta$ contains the disc $Z_4=\{z\in\C:|z+\la_1(\theta)/2|\le \la_1(\theta)/2\}$.
%$$\left\{\la\in\C: \Re\la\le -\dfrac{\la_1(\theta)}2, |\Im\la|\le \dfrac{\la_1(\theta)}2 \right\} \cup 
%\left\{\la\in\C:
%\left|\la+\dfrac{\la_1(\theta)}2\right|\le \dfrac{\la_1(\theta)}2 \right\}\subset \Sigma_\theta.$$
\vspace{2mm}
\item[{\rm (iii)}] $\,\|\alpha\|^2 + \dfrac{\|\al\|\,\|\alpha'\|}{2\sqrt{\lambda_1(\theta)}} > \lambda_1(\infty)+ \la_1(\theta)${$:$} \vspace{-1mm} Then
$$
  \la_1(\theta) < a_\theta
$$
and $a_\theta$ is the $($unique$)$ solution of the equation
\begin{equation}
\label{eq}
  2 \, \sqrt{\la} \, \big(\la + \la_1(\infty) - \|\al\|^2 \big)  = \|\al\|\,\|\al'\| \ \text{ in } \ (\la_1(\theta),\infty);
\end{equation}
moreover, $\Sigma_\theta$ contains the disc $Z_3 \cup Z_4 \cup Z_5 = \{z\in\C:|z|\le \la_1(\theta)\}$. \vspace{1mm}
%$$\{\la\in\C: \Re\la\le 0, |\Im\la|\le \la_1(\theta)\} \cup \{ \la\in \C: \Re \la >0, |\la|\le \la_1(\theta)\}\subset \Sigma_\theta.$$
\end{hlist}
\end{theorem}

In Figure \ref{(i)} the boundary of the set $\Sigma_\theta$ containing the spectrum of the dynamo operator is displayed for the three cases above and $l=1$, keeping the colour scheme for the sets $Z_i$ from Figure \ref{Z}.

\begin{figure}[h]
\vspace{1mm}
\parbox{6cm}
{\epsfysize=6cm
\epsfbox{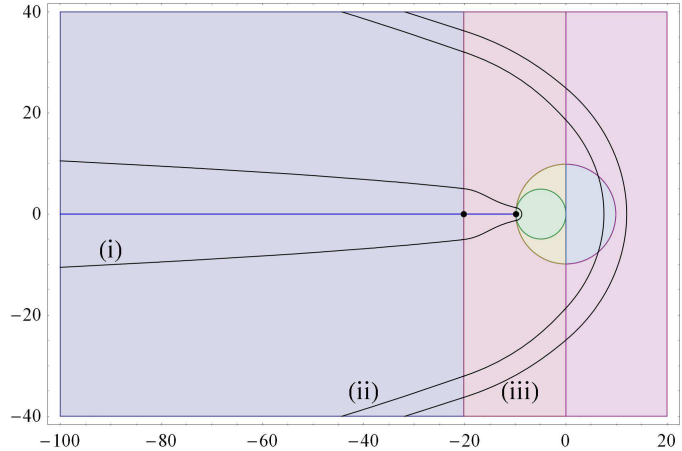}}
% \parbox{6cm}
% {\epsfysize=6cm
% \epsfbox{a_1__ap_1.eps}}
% \parbox{6cm}
% {\epsfysize=6cm
% \epsfbox{a_5__ap_3.eps}}
% \parbox{6cm}
% {\epsfysize=6cm
% \epsfbox{a_5__ap_10.eps}}
\caption{\label{(i)} Boundary of $\Sigma_\theta$ in Theorem \ref{!!!} (i), (ii), (iii) for $l=1$.}
\end{figure} 

\begin{remark}
The second inequality in case (ii) can also be written as 
\[
  \|\alpha\|^2 \!+ \dfrac{\|\al\|\,\|\alpha'\|}{\sqrt{\lambda_1(\theta)}} \le \la_1(\infty) + \la_1(\theta) + \dfrac{\|\al\|\,\|\alpha'\|}{2\sqrt{\lambda_1(\theta)}}. 
\]
This shows that case (ii) indeed appears and that (i), (ii), and (iii) exhaust all possible cases for $\|\al\|$ and $\|\al'\|$.
\end{remark}

\begin{corollary}
\label{stable} 
The operator $\CA_\theta$ has no spectrum in the closed right half-plane~if 
\[
  %\left(\|\alpha\|^2 + \dfrac{\|\al\|\|\alpha'\|}{\sqrt{\lambda_1(\theta)}}\right) \frac 1{\lambda_1(\infty)} <1.
  \|\alpha\|^2 + \dfrac{\|\al\|\|\alpha'\|}{\sqrt{\lambda_1(\theta)}} < \lambda_1(\infty).
\]
\end{corollary}

{\bf Proof of Theorem \ref{!!!}.}
The existence of the function $h_\theta$ with the claimed properties follows from Proposition \ref{thm2}, Lemma \ref{implicit} and the Implicit Function Theorem, applied to the restriction of $\varphi$ to the open upper half-plane.

The monotonicity properties of $\varphi$ in Lemma \ref{implicit} induce corresponding monotonicity properties for the function $f$;
in particular, $f$ is strictly decreasing on $(-\la_1(\theta),\infty)$. Therefore, since
\[
  \lim_{t\searrow-\la_1(\theta)}f(t)\, =\lim_{t\searrow-\la_1(\theta)}f_4(t)\, =\, \infty, \quad 
  \lim_{t\to\infty}f(t)=\lim_{t\to\infty}f_6(t)=0,
\]
the equation $f(\la)=1$ has exactly one solution in $(-\la_1(\theta),\infty)$; this solution is the 
unique zero $a_\theta$ of $h_\theta$ and hence equal to $\max \Re \Sigma_\theta$.
The location of $a_\theta$ is classified by the three cases in Theorem \ref{!!!}:
\begin{hlist}
\item[{\rm (i)}] The condition in (i) is equivalent to $f(0)=f_4(0)\le 1$ and thus $a_\theta \in (-\la_1(\theta),0]$.  Now the formula for $a_\theta$ is obtained by solving the quadratic equation 
\[
  \big(a_\theta + \la_1(\infty) \big) \big(a_\theta+\la_1(\theta)\big) = \|\al\|^2 \la_1(\theta) + \|\al\|\, \|\al'\| \sqrt{\la_1(\theta)}
  %\hspace*{-6mm}
\] 
for $a_\theta$ in the interval $(-\la_1(\theta),0]$, which is equivalent to $f_5(a_\theta) =1$. 
\vspace{1mm}
\item[{\rm (ii)}] The first condition in (ii) is equivalent to $f(0)=f_5(0)>1$, while the second condition in (ii) is equivalent to $f_5(\la_1(\theta))\le 1$. In this case $a_\theta\in(0,\la_1(\theta)]$. The formula for $a_\theta$ is obtained by solving the quadratic equation 
\[
  \big(a_\theta + \la_1(\infty)-\|\al\|^2 \big) \big(a_\theta+\la_1(\theta)\big) = \|\al\|\, \|\al'\| \sqrt{\la_1(\theta)}
  %\hspace{-6mm}
\]
for $a_\theta$ in the interval $(0,\la_1(\theta)]$, which is equivalent to $f_4(a_\theta) =1$.
\\
Moreover,  for $\la\in Z_4 = \{ z\in \C : -\la_1(\theta) < \Re z \le 0$, $|z+\la_1(\theta)/2)| \le \la_1(\theta)/2\}$, we have the estimates 
$|\la+\la_1(\theta)| \le \la_1(\theta)$, $|\la+\la_1(\theta)| \le \la_1(\infty)$, and hence 
\begin{align*}
 f_4(\la) &\ge \left( \!\|\al\|^2 \!\!+ \frac{\|\al\|\,\|\al'\|}{\sqrt{\la_1(\theta)}} \right) \frac{1}{\la_1(\infty)} >1, \quad \la\in Z_4,
\end{align*}
by the first condition in case (ii). This proves that $Z_4 \subset \Sigma_\theta$.
%From the property $f(\la_1(\theta)/2+\,\i \la_1(\theta)/2) =f_4(\la_1(\theta)/2+\,\i \la_1(\theta)/2) >1$ and the monotonicity properties of the function $\varphi$ induced by $f$ (see Lemma \ref{implicit}) it follows that then also $f(\la) >1$ for all $\la \in \C$ with $\Re\la \le -\la_1(\theta)/2$ and  $0< \Im \la \le \la_1(\theta)/2$. \vspace{1mm}
\item[{\rm (iii)}] The condition in (iii) is equivalent to $f(\la_1(\theta))=f_6(\la_1(\theta))> 1$. Therefore $a_\theta\in(\la_1(\theta),\infty)$
and the equation $f(\la)=1$ is equivalent to $f_6(\la) =1$, which is~\eqref{eq}.
\\
Furthermore, for $\la \in Z_5 = \{ z\in\C: \Re z > 0, |z| \le \la_1(\theta) \}$, we have the estimates $|\la+\la_1(\theta)| \le 2 \la_1(\theta)$,
$|\la+\la_1(\theta)| \le \la_1(\infty)+\la_1(\theta)$, and hence 
\begin{align*}
 f_5(\la) & \ge \left( \|\al\|^2 + \|\al\| \, \|\al'\| \frac 1 {2\sqrt{\la_1(\theta)}} \right) \frac 1{\la_1(\infty)+\la_1(\theta)} > 1, \quad \la\in Z_5,
\end{align*}
by the condition in (iii). This shows that $Z_5 \subset \Sigma_\theta$. From the monotonicity properties of $f$ it follows that also $Z_3 \cup Z_4 \subset \Sigma_\theta$.
%From the property $f(\,\i \la_1(\theta)) =f_5(\,\i\la_1(\theta)) >1$ and the monotonicity properties of $f$, it follows that then also $f(\la) >1$ for all $\la \in %\C$ with $\Re\la \le 0$ and 
%$0< \Im \la \le \la_1(\theta)$. 
\qed
\end{hlist}

\smallskip

\begin{remark}
The spectral inclusion in Theorem~\ref{!!!} does not provide a uniform bound for the imaginary parts of the eigenvalues in the left half-plane.
Such a bound is obtained in the next section where we show that the imaginary parts of all eigenvalues are bounded by $\|\al'\|$ (see Theorem \ref{strip}).
\end{remark}

While Proposition~\ref{thm1} was applied in the proof of Theorem \ref{!!!} to establish a global inclusion for all eigenvalues of $\CA_\theta$,
it can also be used to study the local behaviour of eigenvalues.

The eigenvalues of $\CA_\theta$ for $\alpha \equiv 0$ coincide with the eigenvalues of the diagonal elements 
$-A_\theta$ and~$-A_\infty$ and are hence real; their multiplicity is 1 for $\theta \in [0,\infty)$ and $2$ for $\theta=\infty$.
The following proposition provides a condition guaranteeing that, e.g., if $\theta \in [0,\infty)$, for $\alpha \not\equiv 0$ an eigenvalue remains on the real axis.
With regard to the dynamo problem \eqref{diffsyst}, \eqref{boundcond}, such a results could be called a local non-oscillation theorem.

\begin{proposition}
\label{local}
Let $\theta \in [0,\infty]$ and $\lambda_0\in\sigma(-A_\theta)\cup \sigma(-A_\infty)$. Set
$$
  \delta_0:= \dfrac 12 \, \dist \Big( \lambda_0, \big( \sigma(-A_\theta)\cup \sigma(-A_\infty) \big) \setminus\{\lambda_0\}\Big)
$$ 
and denote by $\Gamma_0$ the circle centered at $\lambda_0$ with radius $\delta_0$. If
\begin{equation}
\label{schp}
   \|\alpha\|^2+\dfrac{\|\alpha\|\,\|\alpha'\|}{\sqrt{|\lambda_0|+2\delta_0}}<\dfrac{\delta_0^2}{|\lambda_0|+2\delta_0},
\end{equation}
then, for $\theta\in [0,\infty)$, the operator $\mathcal A_\theta$ has exactly one eigenvalue within the circle $\Gamma_0$, and this eigenvalue is simple and real;
the operator $\CA_\infty$ has exactly two eigenvalues within the circle $\Gamma_0$ $($counted with multiplicities$)$, and these eigenvalues are either real or form a complex conjugate pair.
\end{proposition}

\begin{proof}
For $\la\in\Gamma_0$, we have (see \eqref{kato})
\begin{align*}
  \max_{\lambda\in\Gamma_0}\,\|(A_\theta + \lambda)^{-1}\| &= \dfrac 1{\delta_0}, \\
  \max_{\lambda\in\Gamma_0}\,\|A_\theta (A_\theta + \lambda)^{-1}\|
  &=\max_{\lambda\in\Gamma_0}\,\|I-\lambda(A_\theta + \lambda)^{-1}\|\le 1+\dfrac{|\lambda_0|+\delta_0}{\delta_0},\\
  \max_{\lambda\in\Gamma_0}\,\|A_\theta^{1/2} (A_\theta + \lambda)^{-1}\|
  &=\max_{\lambda\in\Gamma_0}\,\|A_\theta (A_\theta + \lambda)^{-2}\|^{1/2}\\
  &\le\max_{\lambda\in\Gamma_0}\,\big(\|(A_\theta + \lambda)^{-1}\|+|\lambda|\,\|(A_\theta + \lambda)^{-2}\|\big)^{1/2} \\
  &\le\dfrac{\sqrt{|\lambda_0|+2\delta_0}}{\delta_0}.
\end{align*}
Using these estimates, it is not difficult to see that condition \eqref{schp} guarantees that condition \eqref{important}, and hence condition \eqref{neumann-new} in Proposition~\ref{thm1}, holds for all $\la\in\Gamma_0$. Now Proposition~\ref{thm1} yields that $\Gamma_0 \subset \rho(\CA_\theta)$.

If we introduce the operators $\CA_\theta^{(\varepsilon)}:=\CQ_\theta +\varepsilon \CR$ for $0\le \varepsilon\le 1$, then the same arguments as above yield the inclusion $\Gamma_0\subset\rho(\CA_\theta^{(\varepsilon)})$ for all $0\le \varepsilon\le 1$. Since the eigenprojection
\[
  P(\varepsilon) = - \frac 1{2\pi\i} \int_{\Gamma_0} \big( \CA_\theta^{(\varepsilon)} - z \big)^{-1} \d z
\]
onto the part of the spectrum of $\CA_\theta^{(\varepsilon)}$ inside $\Gamma_0$ depends continuously on $\varepsilon$, the dimension of its range is constant for $0\le \varepsilon\le 1$ (see \cite[Lemma~I.4.10, Theorem~IV.3.16 and its proof]{MR0407617}). For $\theta\in [0,\infty)$, the operator $\CA_\theta^{(0)}$, and hence every operator $\CA_\theta^{(\varepsilon)}$ for $0\le \varepsilon\le 1$, has exactly one eigenvalue of multiplicity $1$ inside~$\Gamma_0$. The operator $\CA_\infty^{(0)}$,  has exactly one eigenvalue of multiplicity $2$ inside~$\Gamma_0$ and hence every operator $\CA_\infty^{(\varepsilon)}$ for $0\le \varepsilon\le 1$ has exactly two eigenvalues  in $\Gamma_0$ counted with multiplicities.

Since all entries of the operator matrices $\CA_\theta^{(\varepsilon)}$ are differential operators with real coefficients, the spectra $\sigma(\CA_\theta^\varepsilon)$ are symmetric to $\R$ for all $0\le \varepsilon\le 1$ (compare Theorem \ref{harz}). Thus, for $\theta\in [0,\infty)$, the single eigenvalue of 
$\CA_\theta^{(1)}=\CA_\theta$ inside~$\Gamma_0$ must be real, while $\CA_\infty^{(1)}=\CA_\infty$ has two eigenvalues inside~$\Gamma_0$ counted with multiplicities which are either real or form a complex conjugate pair.
\end{proof}

\begin{corollary}
\label{meet}
Let $\theta \in [0,\infty)$ and let $\Gamma_0$ be the circle around the largest eigenvalue $-\la_1(\theta)$ of the diagonal elements of $\CA_\theta$ %(i.e.\ of $\CA_\theta$ in the case $\alpha\equiv 0$) 
with radius $\delta_0=\big( \la_1(\infty)-\la_1(\theta) \big) /2$. If
\begin{equation}
\label{meet-cond}
 \|\al\|^2 + \frac{\|\al\|\,\|\al'\|}{\sqrt{\la_1(\infty)}} < \frac {\big(\la_1(\infty)-\la_1(\theta) \big)^2}{4\la_1(\infty)}, 
\end{equation}
then the operator $\CA_\theta$  has exactly one real eigenvalue in $\Gamma_0$.
\end{corollary}

\section{A similarity transformation and a second eigenvalue estimate} 
\label{section4}

Our second approach to study the spectral properties of the operator $\CA_\theta$ is based on a quasi-similarity transformation of $\CA_\theta$.
The transformed operator $\CB_\theta$ turns out to be a bounded perturbation of a self-adjoint operator $\CS_\theta$.
This allows us to prove another estimate for the eigenvalues of $\CA_\theta$ which shows, in particular, that the 
imaginary parts of all eigenvalues have modulus at most $\|\al'\|$. In addition, we investigate the number of positive eigenvalues of the unperturbed operator $\CS_\theta$.

\begin{proposition}
\label{B_theta-sim}
Let $\theta\in [0,\infty]$ and let the linear operator $\CW_\theta$ in $L_2(0,1)\oplus L_2(0,1)$ \vspace{-2mm} be given by
\[
  \CW_\theta := \matrix{cc}{A_\theta^{1/2}& 0 \\ 0 & I}, \quad \CD(\CW_\theta) := \CD\big(A_\theta^{1/2}\big) \oplus L_2(0,1).
\]
Then the linear operator $\CB_\theta$ in $L_2(0,1) \oplus L_2(0,1)$ defined as
\begin{equation}
\label{transform}
  \CB_\theta := \CW_\theta \CA_\theta \CW_\theta^{-1}  
\end{equation}
is closed and the eigenvalues of $\,\CB_\theta$ coincide with those of $\,\CA_\theta$,
\[ 
  \sigma_{\rm p}(\CB_\theta) = \sigma_{\rm p}\big(\CA_\theta\big);
\]
moreover, $(x_j)_{j=0}^k$ is a Jordan chain of $\,\CA_\theta$ at $\la \in \sigma_{\rm p}\big(\CA_\theta\big)$ if and only if
$\,(\CW_\theta x_j)_{j=0}^k$ is  a Jordan chain of $\,\CB_\theta$ at $\la$.
%and if $\lambda\in\sigma_p(\mathcal B_\theta)$ and $\mathcal B_\theta y=\lambda y$ then $\mathcal A_\theta \mathcal W_\theta^{-1}y = \lambda\mathcal W_\theta^{-1}y$.
\end{proposition}

\begin{proof}
Since $\CW_\theta$ is boundedly invertible and $\CA_\theta$ is closed, the product $\CB_\theta$ is a closed operator (see \cite[Section~III.5.2]{MR0407617}). 

In order to prove the equality of the point spectra and the claim about the Jordan chains, we observe that, by definition, the domain of $\CB_\theta$ is given by
\begin{equation}
\label{Bdom}
 \CD(\CB_\theta) = \big\{ y \in L_2(0,1)\oplus L_2(0,1) : \, \CW_\theta^{-1} y \in \CD\big(\CA_\theta\big), \ \CA_\theta \CW_\theta^{-1} y \in \CD(\CW_\theta) \big\}.
\end{equation}
Let $\la \in \sigma_{\rm p}(\CA_\theta)$ and let $(x_j)_{j=0}^k \subset \CD(\CA_\theta)$ be a corresponding Jordan chain, that is, 
$(\CA_\theta - \la) x_j = x_{j-1}$ for $j=0,1,\dots,k$, where $x_{-1}:=0$. Since $\CD(\CA_\theta) \subset \CD(\CW_\theta)$, it follows that $\CA_\theta  x_j = \la_0 x_j + x_{j-1} \in \CD(\CW_\theta)$. Hence $\CW_\theta x_j \in \CD(\CB_\theta)$ by \eqref{Bdom} and 
\[ 
  \big(\CB_\theta - \la \big) \CW_\theta x_j = \big(\CW_\theta \CA_\theta \CW_\theta^{-1} - \la \big) \CW_\theta x_j 
                                      = \big(\CW_\theta \CA_\theta - \la \CW_\theta \big) x_j = \CW_\theta x_{j-1}  
\]
for $j=0,1,\dots, k$. Vice versa, let $\la \in \sigma_{\rm p}(\CB_\theta)$ and let $(y_j)_{j=0}^k \subset \CD(\CB_\theta)$ be a corresponding Jordan chain, that is, 
$(\CB_\theta - \la) y_j = y_{j-1}$ for $j=0,1,\dots,k$, where $y_{-1}:=0$.
Since $\CW_\theta^{-1} y_j  \in \CD(\CA_\theta)$ by \eqref{Bdom}, we obtain
\[
  \big(\CA_\theta - \la \big) \CW_\theta^{-1} y_j = \big(\CW_\theta^{-1} \CB_\theta \CW_\theta - \la \big) \CW_\theta^{-1} y_j =
  \big(\CW_\theta^{-1} \CB_\theta - \la \CW_\theta^{-1} \big) y_j = \CW_\theta^{-1} y_{j-1}
\]
for $j=0,1,\dots, k$.
\end{proof}

In the next theorem we show that the operator $\CB_\theta$ is a bounded perturbation~of a self-adjoint semi-bounded block operator matrix. 
First we need an auxiliary lemma.

\begin{lemma}
\label{diagdom}
Let $\theta \in [0,\infty]$. Then
\begin{enumerate}
 \item[{\rm i)}]  $\alpha \,A_\theta^{1/2}$ is $A_\theta$-compact, 
 \item[{\rm ii)}] $A_\theta^{1/2} \alpha$ is $A_\infty$-compact.
\end{enumerate}
\end{lemma}

\begin{proof}
i) Since $A_\theta$ has compact resolvent by Proposition~\ref{A_theta}~ii), we conclude that $\al A_\theta^{1/2} A_\theta^{-1} = \al A_\theta^{-1/2}$ is compact.

ii) First we show that 
\begin{equation}
\label{half}
  \CD(A_\infty^{1/2}) \subset \CD\big(A_\theta^{1/2} \alpha\big).
\end{equation}
The description of $\CD\big(A_\theta^{1/2}\big)$ in Proposition~\ref{A_theta}~i) and the assumption $\al \!\in\! C^1([0,1])$ 
show that $x\!\in\! \CD\big(A_\theta^{1/2}\big)$ implies $\al x \!\in\! \CD\big(A_\theta^{1/2}\big)$, that~is, $\CD\big(A_\theta^{1/2}\big) \subset \CD\big(A_\theta^{1/2} \alpha\big)$.
Since $\CD(A_\infty^{1/2}) = \big\{ x \in \CD(A_\theta^{1/2}) : x(1) = 0\big\} \subset \CD(A_\theta^{1/2})$ by Proposition \ref{A_theta} i), the inclusion \eqref{half} follows. 

Because $A_\infty^{1/2}$ and $A_\theta^{1/2} \alpha$ are closed, \eqref{half} implies that $A_\theta^{1/2} \al A_\infty^{-1/2}$ is a bounded operator (see \cite[Remark IV.1.5]{MR0407617}).
Since $A_\infty$ has compact resolvent by Proposition~\ref{A_theta}~ii), it follows that $A_\theta^{1/2} \al  A_\infty^{-1} = A_\theta^{1/2} \al  A_\infty^{-1/2} A_\infty^{-1/2}$ is compact.
\end{proof}

\begin{theorem}
\label{B_theta}
Let $\theta\in [0,\infty]$. Then the operator $\CB_\theta$ defined in \eqref{transform} has the form
\begin{equation}
\label{B0}
  \CB_\theta = \matrix{cc}{ - A_\theta & A_\theta^{1/2} \al \\ A_{\theta,\al} A_\theta^{-1/2} & - A_\infty}, \quad 
  \CD(\CB_\theta) = \CD\big(A_\theta\big) \oplus \CD\big(A_\infty\big). 
\end{equation}
If we define
\begin{alignat}{2}
\label{B3}   
  \mathcal \CS_\theta & := \matrix{cc}{-A_\theta & A_\theta^{1/2} \alpha  \\  \alpha A_\theta^{1/2} & - A_\infty},
  & \quad \CD(\CS_\theta) & := \CD\big(A_\theta\big) \oplus \CD\big(A_\infty\big), \\
\label{B2} 
  \mathcal \CT_\theta &:= \matrix{cc}{0 & 0 \\ -\alpha' D A_\theta^{-1/2} & 0},
  & \quad \CD(\CT_\theta) & := L_2(0,1) \oplus L_2(0,1),
\end{alignat}
then $\CS_\theta$ is self-adjoint and bounded from above with compact resolvent, $\CT_\theta$ is bounded with $\|\CT_\theta\| \le \|\al'\|$, and 
\begin{equation}
\label{B1}
   \CB_\theta = \CS_\theta + \CT_\theta.
\end{equation}
\end{theorem}

\begin{proof}
First we prove the identity \eqref{B0}.
By the definition  of $A_{\theta,\al}$ in Proposition~\ref{A_theta-alpha}, we have $\CD(A_{\theta,\al}) = \CD(A_\theta)$ and
thus $\CD(A_\theta)  \subset \CD(A_\theta^{1/2}) = \CD(A_{\theta,\al} A_\theta^{-1/2})$; 
by Lemma~\ref{diagdom} ii) (see also \eqref{half}), we have $\CD(A_\infty) \subset \CD\big(A_\theta^{1/2} \alpha\big)$. This shows that
$\CD(A_\theta) \oplus \CD(A_\infty)$ is the domain of the block operator matrix in \eqref{B0}. Formally, the relation \eqref{B0} follows from the definition of $\CB_\theta$ in \eqref{transform}; it remains to be shown that $\CD(\CB_\theta) = \CD\big(A_\theta\big) \oplus \CD\big(A_\infty\big)$. 

According to \eqref{Bdom}, we have $y\in\CD(\CB_\theta)$ if and only if 
$\CW_\theta^{-1} y \in \CD(\CA_\theta)$ and $\CA_\theta \CW_\theta^{-1} y \in \CD(\CW_\theta)$.
Since $\CD(\CA_\theta)=\CD\big(A_\theta\big) \oplus \CD\big(A_\infty\big)$ by \eqref{bomA_l} and 
\begin{equation}
\label{sp}
    A_{\theta,\al} A_\theta^{-1/2} = \alpha A_\theta^{1/2} - \alpha' D A_\theta^{-1/2}
\end{equation}
by \eqref{opid}, it follows that $y=(y_1,y_2)^{\rm t} \in \CD(\CB_\theta)$ if and only~if
\[
 \left\{ 
  \begin{array}{ll} 
  {\rm (a)}  &  A_\theta^{-1/2} y_1 \in \CD(A_\theta), \quad  y_2 \in  \CD(A_\infty),  \\[1mm]
  {\rm (b)} &  - A_\theta^{1/2} y_1 + \alpha \, y_2 \in \CD\big(A_\theta^{1/2}\big), \  \ \big(\alpha A_\theta^{1/2} - \alpha' D A_\theta^{-1/2}\big) y_1 - A_\infty y_2 \in L_2(0,1).
 \end{array}
 \right. 
\]
The first condition in (a) is equivalent to $y_1 \in \CD(A_\theta^{1/2})$. Thus the second condition in (b) always holds 
as $D A_\theta^{-1/2}$ is bounded by Lemma \ref{march3}. 
By Lemma \ref{diagdom} ii) (see also \eqref{half}), we have $\CD(A_\infty) \subset \CD(A_\theta^{1/2} \alpha)$
and hence the first condition in (b) reduces to $A_\theta^{1/2} y_1 \in \CD(A_\theta^{1/2})$. Altogether,
$y=(y_1,y_2)^{\rm t}\in\CD(\CB_\theta)$ if and only if
$y_1\in \CD(A_\theta)$ and $y_2\in \CD(A_\infty)$. This completes the proof of \eqref{B0}.

To see that $\CS_\theta$ is symmetric, we observe that
$\alpha \in C([0,1])$ and  $\alpha$ is real-valued. This implies that the corresponding multiplication operator is bounded and self-adjoint and hence
\[
  \big(\alpha A_\theta^{1/2}\big)^\ast = \big( A_\theta^{1/2} \big)^\ast \alpha^\ast = A_\theta^{1/2} \alpha.  
\]
Furthermore, Lemma \ref{diagdom} implies that
\[
 \CS_\theta = \matrix{cc}{-A_\theta & 0 \\ 0 & -A_\infty} + \matrix{cc}{0 & A_\theta^{1/2} \alpha  \\  \alpha A_\theta^{1/2} & 0}
\]
is a relatively compact perturbation of its block diagonal part ${\rm diag} (-A_\theta,-A_\infty)$. The latter is self-adjoint and semi-bounded (in fact, negative) 
with compact resolvent by Proposition \ref{A_theta}. Hence $\CS_\theta$ has the same properties by the Rellich-Kato Theorem, the stability theorem for semi-boundedness, and 
Weyl's essential spectrum theorem (see \cite[Theorems~V.4.3 and V.4.11]{MR0407617}, \cite[Corollary XIII.4.2]{MR0493421}).

Finally, because $\alpha'\in C([0,1])$ and $D A_\theta^{-1/2}$ is bounded with $\|D A_\theta^{-1/2}\| \le 1$ (see Lemma \ref{march3}), the operator $\CT_\theta$ is bounded with $\|\CT_\theta\| \le \|\al'\|$. 

Since $\CD(\CS_\theta)= \CD(\CB_\theta)$, the identity \eqref{B1} is immediate from \eqref{B0} and~\eqref{sp}.
\end{proof}

\begin{remark}
The boundedness from above of $\CS_\theta$ is proved independently in Proposition \ref{stheta}, where also a concrete upper bound for $\CS_\theta$ is established.
\end{remark}

\begin{proposition}
\label{F_theta}
Let $\theta \in [0,\infty]$. The spectrum $\sigma(\CS_\theta)=\sigma_{\rm p}(\CS_\theta)$ of the self-adjoint operator $\CS_\theta$ 
satisfies
\begin{equation}
\label{schur-spec}
 %\sigma(\CS_\theta) \setminus \sigma(-A_\theta) = \sigma(S_{2,\theta}), \quad
 \sigma_{\rm p}(\CS_\theta) \setminus \sigma_{\rm p}(-A_\theta) = \sigma_{\rm p}(S_{2,\theta})
\end{equation}
where the so-called Schur complement $S_{2,\theta}$ is the operator function given by
\[ 
 S_{2,\theta}(\la):= -A_\infty \!- \la + \al A_\theta^{1/2} (A_\theta+\!\la)^{-1} A_\theta^{1/2} \alpha, \quad \CD\big(S_{2,\theta}(\la)\big) := \CD(A_\infty),
\]
for all $\la \in \C\setminus \sigma(-A_\theta)$. Moreover, the multiplicities of the eigenvalues of $\,\CS_\theta$ and of $S_{2,\theta}$ coincide.
\end{proposition}

\begin{proof}
Let $\la\notin \sigma(-A_\theta)=\sigma_{\rm p}(-A_\theta)$. It is not difficult to check that the block operator matrix $\CS_\theta$ can be factorised as (see, e.g., \cite[Theorem~2.2.14]{book})
\begin{equation}
\label{factorise}
  \CS_\theta \!-\! \la \!=\! \matrix{cc}{I \!& \!\! 0 \\ \!-\al A_\theta^{1/2}(A_\theta\!+\!\la)^{-1}\!&\! I} \!\!
                               \matrix{cc}{\!-A_\theta\!-\!\la \!\!&\!\! 0 \\ 0 \!\!\!&\!\!\! S_{2,\theta}(\la)\!} \!\!
                               \matrix{cc}{I \!&  \!-(A_\theta\!+\!\la)^{-1}\!A_\theta^{1/2}\!\alpha \!\\ 0 \!&\! I}\!. 
  \hspace*{-2mm}
\end{equation}
Since the outer two factors are bounded and boundedly invertible and so is $-A_\theta\!-\!\la$ by the assumption on $\la$, it follows that $\CS_\theta - \la$ 
is boundedly invertible or injective if and only if so is the operator $S_{2,\theta}(\la)$ (see, e.g., \cite[Theorem~2.3.3]{book}). This proves \eqref{schur-spec}. 

% superfluous  now since corners relatively compact yield simpler proof!
% Moreover, for 
% $\la \in  \rho(\CS_\theta) \setminus \sigma(-A_\theta) = \rho(S_{2,\theta})$, we see that
% \begin{align*}
%  (\CS_\theta \!- \la)^{-1} \!=\! \matrix{cc}{I \!&\!\! (A_\theta\!+\!\la)^{-1}A_\theta^{1/2}\alpha \!\\ 0 & I} \!\!
%                                \matrix{cc}{\!-(A_\theta\!+\!\la)^{-1} \!\!\!\!\!&\!\!\!\! 0 \\ 0 \!\!\!&\!\!\! S_{2,\theta}(\la)^{-1}\!\!} \!\!
%                                \matrix{cc}{I & 0 \\ \!\al A_\theta^{1/2} (A_\theta\!+\!\la)^{-1}\!\!&\! I}\!.
% \end{align*}
% By Proposition \ref{A_theta} ii), both operators $A_\theta$ and $A_\infty$ have compact resolvents. Hence 
% \[
%   S_{2,\theta}(\la_0)^{-1} =  -(A_\infty + \la_0)^{-1} \big( I -\al A_\theta^{1/2} (A_\theta+\la_0)^{-1} A_\theta^{1/2} \alpha (A_\infty + \la_0)^{-1} \big)^{-1}
% \]
% is compact for $\la_0 \!\in\! (-\la_1(\theta),\infty)$ with $\big\|\al A_\theta^{1/2} (A_\theta+\la_0)^{-1} A_\theta^{1/2} \alpha (A_\infty \!+ \la_0)^{-1}\big\|<1$ (in particular, for $\la_0$ sufficiently large), and hence for every other $\la\notin \sigma(-A_\theta)$. Thus $(\CS_\theta - \la)^{-1}$ is compact for every $\la\in\rho(\CS_\theta) \setminus \sigma(-A_\theta)= \rho(S_{2,\theta})$.

It is easy to see that $y_0 =(y_{0,1}, y_{0,2})^{\rm t} \in \CD(\CS_\theta)$ is an eigenvector of $\CS_\theta$ at 
$\la_0 \in \sigma_{\rm p}(\CS_\theta)$ if and only if $y_{0,2} \in \CD(A_\infty)$ is an eigenvector of $S_{2,\theta}$ at $\la_0$, that is, $S_{2,\theta}(\la_0)y_{0,2} = 0$, and $y_{0,1}=(A_\theta+\la)^{-1} A_\theta^{1/2}\alpha \,y_{0,2}$. 
Clearly, the self-adjoint operator $\CS_\theta$ has no associated vectors.
Since $S_{2,\theta}$ is a self-adjoint operator function, that is, $S_{2,\theta}(\la)=S_{2,\theta}(\ov\la)^\ast$ for 
$\la\in \C\setminus \sigma(-A_\theta)$, and its derivative $$S_{2,\theta}'(\la)= - I - \al A_\theta^{1/2} (A_\theta+\la)^{-2} A_\theta^{1/2} \alpha \le -I$$ is uniformly negative for all $\la\in \C\setminus \sigma(-A_\theta)$, $S_{2,\theta}$ has no associated vectors either (compare \cite[Lemma~30.13]{MR971506}).
\end{proof}

% superfluous  now since corners relatively compact yield simpler proof!
% \begin{remark}
% That $\CS_\theta$ has compact resolvent can also be proved by means of the relation $\CS_\theta - \la = \CW (\CA_\theta - \la) \CW^{-1} - \CR$ and the fact that $\CA_\theta$ has compact resolvent.
% Here one has to note that, for all $\la\in\C$, the product $\CW (\CA_\theta - \la) \CW^{-1}$ is a Fredholm operator since $\CA_\theta$ and $\CW$ are densely defined Fredholm operators (see, e.g., \cite[Theorem~XVII.3.1]{MR1130394}). Moreover, one has to show that $\CW (\CA_\theta - \la)^{-1}$ is compact; this implies that $\CR \CW (\CA_\theta - \la)^{-1} \CW^{-1}$ is compact and hence 
% $\CS_\theta-\la$ is a Fredholm operator for all $\la\in\C$ (see, e.g., (see \cite[Theorem~XVII.4.3]{MR1130394}).
% \end{remark}

%note that $\CS_\theta$ is self-adjoint and hence closed and so is $S_{2,\theta}(\mu)$ for $\mu\!\in\!(-\!\la_1(\theta),\infty)$.

% \begin{remark} 
%  It is not difficult to see that $\CS_\theta$ can be factorised as
% \begin{align*}
%  \CS_\theta = \matrix{cc}{ I & ) \\ -A_\theta^{-1/2} \al 0 & I} \matrix{cc}{-A_\theta & 0 \\ 0 &  -A_\infty\!+\al^2} \matrix{cc}{I & -\al A_\theta^{-1/2} \\ 0 & I}.
% \end{align*}
% \end{remark}

Next we establish an explicit upper bound for the self-adjoint operator $\CS_\theta$. 
To this end, we first show that $\CS_\theta$ satisfies the following two-sided estimate, which might be of independent interest.

\begin{lemma}
Let $\theta \in [0,\infty]$ and let $\gamma \in (0,1]$ be arbitrary. Then 
the operator $\CS_\theta$ admits the two-sided estimate
\begin{equation}
\label{ie}
 \begin{pmatrix}
  -\big(1+\gamma\big) A_\theta&0\\[2mm]
  0&-A_\infty\!-\!\dfrac 1\gamma \alpha^2
  \end{pmatrix}
  \le\CS_\theta\le\begin{pmatrix}
  -\big(1-\gamma\big) A_\theta&0\\[2mm]
  0&-A_\infty\!+\!\dfrac 1\gamma \alpha^2
  \end{pmatrix};
\end{equation}
in particular, for $\gamma =1$, 
\begin{equation}
\label{ZB}
  \begin{pmatrix}
  -2A_\theta&0\\[2mm]
  0&-A_\infty-\alpha^2
  \end{pmatrix}
  \le\CS_\theta\le\begin{pmatrix}
  0&0\\[2mm]
  0&-A_\infty+\alpha^2
\end{pmatrix}.
\end{equation}
\end{lemma}

\begin{proof}
Let $y =(y_1,y_2)^{\rm t} \in \CD(\CS_\theta) = \CD(A_\theta) \oplus \CD(A_\infty)$.
Then the right in\-equality in \eqref{ie} follows from the estimate
\begin{align*}
\left(\CS_\theta y,y
\right)&=-(A_\theta y_1,y_1) - (A_\infty y_2,y_2) + 2\,\Re\big(A_\theta^{1/2}\alpha \, y_2,y_1\big)\\
&\le -(A_\theta y_1,y_1) - (A_\infty y_2,y_2) + 2 \, \gamma^{1/2} \big\|A_\theta^{1/2}y_1\big\| \frac 1{\gamma^{1/2}} \|\alpha y_2\| \\
&\le -(A_\theta y_1,y_1) - (A_\infty y_2,y_2) + \gamma(A_\theta y_1,y_1)+  \frac 1\gamma \big(\alpha^2y_2,y_2\big);
\end{align*}
the left inequality in \eqref{ie} is obtained analogously.
\end{proof}

\begin{proposition}
\label{alpha-const}
If $\alpha$ is constant, $\alpha \equiv \alpha_0$, we also have the estimate
\begin{equation}
\label{const-al}
  \max\sigma(\CA_\theta) = \max\sigma(\CS_\theta) \le \frac{|\al_0|^2}4.
\end{equation} 
\end{proposition}

\begin{proof}
Let $y=(y_1,y_2)^{\rm t}\!\in \CD(\CS_\theta) = \CD(A_\theta) \oplus \CD(A_\infty)$; then $y_2 \in \CD(A_\infty) \subset \CD(A_\infty^{1/2}) \subset  \CD(A_\theta^{1/2})$ by Proposition \ref{A_theta} i). 
Since $\alpha$ is constant, it commutes with~$A_\theta^{1/2}\!$.~Hence
\begin{align*}
   \left(\CS_\theta y,y \right)%=&-(A_\theta y_1,y_1) - (A_\infty y_2,y_2) 
                                 %       + \big(A_\theta^{1/2}\al_0 \, y_2,y_1\big) + \big(\al_0 A_\theta^{1/2}\, y_1,y_2\big)\\[1mm]
                                      =&-(A_\theta y_1,y_1) - (A_\infty y_2,y_2) 
                                        + \big(\al_0 A_\theta^{1/2} y_2,y_1\big) + \big(\al_0 A_\theta^{1/2} y_1,y_2\big)\\[1mm]
                                      \le& - (A_\theta y_1,y_1) + |\al_0| \sqrt{(A_\theta y_1,y_1)} \|y_2\| 
                                           - (A_\infty y_2,y_2)  + |\al_0| \,\|A_\theta^{1/2} y_2\| \,\|y_1\| \\
                                      \le& - \left( \sqrt{(A_\theta y_1,y_1)} - \frac{|\al_0|}2 \|y_2\| \right)^2 +\frac{|\al_0|^2}4 \|y_2\|^2\\
                                         & - \left( \sqrt{(A_\infty y_2,y_2)} - \frac{|\al_0|}2 \frac{\|A_\theta^{1/2} y_2\|}{\|A_\infty^{1/2} y_2\|} \|y_1\| \right)^2
                                            +\frac{|\al_0|^2}4 \frac{\|A_\theta^{1/2} y_2\|^2}{\|A_\infty^{1/2} y_2\|^2} \|y_1\|^2\\
                                      \le& \,\frac{|\al_0|^2}4 \big( \|y_1\|^2 + \|y_2 \|^2 \big) =  \frac{|\al_0|^2}4 \|y\|^2
\end{align*}
if we show that
\[
 \|A_\theta^{1/2} y_2\| \le \|A_\infty^{1/2} y_2\|, \quad y_2 \in \CD(A_\infty), 
\]
or, equivalently, $\|A_\theta^{1/2} A_\infty^{-1/2}\|\le 1$. 
This follows from the relation $A_\theta^{1/2} A_\infty^{-1/2}=(A_\infty^{-1/2} A_\theta^{1/2})^*$ 
(note that $A_\theta^{1/2} A_\infty^{-1/2}$ is bounded as $\CD(A_\infty^{1/2}) \subset  \CD(A_\theta^{1/2})$)
and from the estimate
\[
 \| A_\infty^{-1/2} A_\theta^{1/2}  x \|^2 = (A_\infty^{-1} A_\theta^{1/2}  x, A_\theta^{1/2}  x ) \le  (A_\theta^{-1} A_\theta^{1/2}  x, A_\theta^{1/2}  x ) = \|x\|^2,
   \quad x\in \CD(A_\theta^{1/2}),
\]
where we have used the inequality $A_\infty^{-1} \le A_\theta^{-1}$ from Pro\-posi\-tion \ref{A_theta} iv). 

Since $\alpha$ is constant and hence $\alpha'\equiv 0$, we have $\CB_\theta=\CS_\theta$ and so \eqref{const-al} follows from Proposition \ref{B_theta-sim}.
\end{proof}

The following proposition provides an upper bound for the unperturbed operator $\CS_\theta$ for arbitrary functions $\alpha$.

\begin{proposition}
\label{stheta}
Let $\theta \in [0,\infty]$. Then the self-adjoint operator $\CS_\theta$ defined in \eqref{B2} satisfies the estimate 
$$\max \sigma(\CS_\theta) \le s_\theta$$ 
where $s_\theta \in [-\la_1(\theta),-\la_1(\infty) + \|\al\|^2]$ is given by
\[
  s_\theta :=
  \left\{ \begin{array}{ll}
  %\!\!-\la_1(\theta) + \|\al\| \sqrt{\la_1(\theta)} \  \displaystyle{\tan \bigg( \frac 12 \arctan \frac{2 \|\al\| \sqrt{\la_1(\theta)}} {\la_1(\infty) - \la_1(\theta)} \bigg)}
  \!\!\displaystyle{- \frac{\la_1(\infty)+\la_1(\theta)}2 + \sqrt{ \bigg( \frac{\la_1(\infty) - \la_1(\theta)}2\bigg)^2 + \la_1(\theta) \|\al\|^2}}
                                    & \!\text{if }\ \|\al\|^2 \!\le\! \la_1(\infty), \\[1mm]
  \!\!-\la_1(\infty) + \|\al\|^2 & \!\text{if } \ \|\al\|^2 \!>\! \la_1(\infty).
  \end{array} \right.
\]
\end{proposition}

\smallskip

\begin{remark}
\label{siginf}
\begin{hlist}
\item[i)] The bound $s_\theta$ satisfies 
\begin{alignat*}{3}
&\!\!\!-\la_1(\theta) \le s_\theta \le 0 \  & & \iff \ & \ \|\al\|^2 & \le \la_1(\infty), \\
& 0 < s_\theta    \le  -\la_1(\infty) + \|\al\|^2  \ & & \iff \ & \ \|\al\|^2 & > \la_1(\infty).
\end{alignat*}
\item[ii)] If $\|\al\|^2 \!<\! \la_1(\infty)$, then $s_\theta$ can be written equivalently~as
\[
 s_\theta = -\la_1(\theta) + \|\al\| \sqrt{\la_1(\theta)} \  \displaystyle{\tan \bigg( \frac 12 \arctan \frac{2 \|\al\| \sqrt{\la_1(\theta)}} {\la_1(\infty) - \la_1(\theta)} \bigg)}.
\]
\item[iii)] In the particular case $\theta=\infty$, the expression for $s_\infty$  simplifies to
\begin{align*}
  s_\infty & =
  \left\{ \begin{array}{ll}
  \!\!-\la_1(\infty) + \|\al\| \sqrt{\la_1(\infty)} & \text{if } \ \|\al\|^2 \!\le\! \la_1(\infty), \\[1mm]   
  \!\!-\la_1(\infty) + \|\al\|^2 & \text{if } \ \|\al\|^2 \!>\! \la_1(\infty).
  \end{array} \right. \\
  & = - \la_1(\infty) + \|\al\| \max \big\{ \|\al\|, \sqrt{\la_1(\infty)}\,\big\}.
\end{align*}
\end{hlist}
\end{remark}

\noindent
{\it Proof of Proposition} \ref{stheta}.
%To estimate the upper bound of the spectrum, it is sufficient to show that $(\CS_\theta y,y) \le s_\theta \|y\|^2 $ for all $y\in \CD(\CS_\theta)$ (see, e.g., 
%\cite[Theorem~V.3.2]{MR0407617}). 
Since $-A_\theta \!\le\! -\la_1(\theta)$ and  $-A_\infty \!\le\! -\la_1(\infty)$, the right in\-equality in \eqref{ie} yields that, for arbitrary $\gamma\in (0,1]$,
\begin{align}
  \big(\CS_\theta y,  y \big)
  & \le - (1-\gamma) \,\la_1(\theta) \, \|y_1\|^2 + \Big( -\la_1(\infty) + \frac 1 \gamma \|\al\|^2 \Big) \,\|y_2\|^2  \notag \\
  & \le \max \big\{ h_1(\gamma),h_2(\gamma) \big\} \, \| y \|^2 \notag
\end{align}
where we have set 
\[
  h_1(\gamma) :=  - (1-\gamma) \,\la_1(\theta) , \quad h_2(\gamma) :=  -\la_1(\infty) + \frac 1 \gamma \|\al\|^2, \quad \gamma \in (0,1].
\]
If $\|\al\|^2 \ge \la_1(\infty)$, it is not difficult to see that $h_2(\gamma) \ge 0 \ge h_1(\gamma)$ for all $\gamma \in (0,1]$; in this case, the optimal estimate is obtained
for $\gamma=1$, that is,
\[
 \big(\CS_\theta y,  y \big) \le \big( -\la_1(\infty) + \|\al\|^2 \big) \, \|y\|^2.
\]
If $\|\al\|^2 < \la_1(\infty)$, a short calculation shows that the function $\max\{ h_1(\gamma),h_2(\gamma)\}$ attains its minimum at the point $\gamma_0 \in (0,1]$ 
where $h_1$ and $h_2$ intersect and which is given by the relation
\[
  \gamma_0= \frac 1{\la_1(\theta)} \,\bigg(  - \frac{\la_1(\infty)-\la_1(\theta)}2 + \sqrt{ \bigg( \frac{\la_1(\infty) - \la_1(\theta)}2\bigg)^2 + \la_1(\theta) \|\al\|^2}  \bigg);
\]
in this case, the optimal estimate becomes
\begin{align*}
 \big(\CS_\theta y,y \big) \!&\le\! \bigg(\!\! - \frac{\la_1(\infty)+\la_1(\theta)}2 + \sqrt{ \bigg( \frac{\la_1(\infty) - \la_1(\theta)}2\bigg)^2 + \la_1(\theta) \|\al\|^2}  \bigg)  \|y\|^2,%\\
                             %&=\!\bigg(\!\!\!-\min \{\la_1(\infty),\la_1(\theta)\} \! +\! \|\al\| \! \sqrt{\!\la_1(\theta)} 
                             %           \tan \bigg( \!\frac 12 \!\arctan \frac{2 \|\al\|\!\sqrt{\!\la_1(\theta)}} {\la_1(\infty) \!-\! \la_1(\theta)} \bigg) \!\bigg) \|y\|^2\\[1mm]
                             %&=\!\bigg( \!\!-\la_1(\theta) + \sqrt{\la_1(\theta)} \, \|\al\| 
                                        %\tan \bigg( \frac 12 \arctan \frac{2 \sqrt{\la_1(\theta)}\|\al\|} {\la_1(\infty) \!-\! \la_1(\theta)} %\bigg) \!\bigg) \|y\|^2. 
                             %&= s_\theta \, \|y\|^2,
\end{align*}
where we have used that $\la_1(\theta) \le \la_1(\infty)$ for all $\theta \in [0,\infty]$ by Proposition \ref{A_theta} iii).
\qed

\vspace{2mm}

While the first spectral enclosure in Theorem \ref{!!!} only provides some region in the complex plane containing the eigenvalues, 
the next theorem gives more detailed information in terms of the eigenvalues of the self-adjoint operator $\CS_\theta$.

\begin{theorem}
\label{strip}
Let $\theta\in[0,\infty]$. The eigenvalues of $\CA_\theta$ lie in discs of radius $\|\alpha'\|$ around the eigenvalues of the self-adjoint operator $\CS_\theta$:
\begin{equation}
\label{inclusion}
  \sigma(\CA_\theta) \subset \big\{ \la\in \C : \dist\big(\la, \sigma_{\rm p}(\CS_\theta)\big) \le \|\alpha'\| \big\}.
\end{equation}
%The imaginary parts of all the eigenvalues of the dynamo operator $\CA_\theta$ are of modulus $\le \|\al'\|$, 
%the real parts are bounded from above by $\sigma_\theta + \|\alpha'\|$, that~is,
As a consequence,
\begin{equation}
\label{inclusion2}
  \sigma(\CA_\theta) \subset  %\Lambda_\theta:= \big\{ \la \in \C : \Re \la \le b_\theta, \, |\Im\la| \le g_\theta(\Re\la) \big\},
  \big\{ \la \!\in\! \C : \Re \la \le s_\theta, \,  |\Im \la | \le \|\alpha'\| \big\}
  \cup  \big\{ \la \!\in\! \C : | \la - s_\theta| \le  \|\alpha'\| \big\}
\end{equation}  
%where %$b_\theta$ and the function $g_\theta:(-\infty,b_\theta]\to[0,\infty)$ are given by
%\begin{equation}
%\label{btheta}
%  b_\theta := s_\theta+\|\al'\|, \quad g_\theta(t) := \left\{ \begin{array}{cl} \|\al'\|, \ &  t\in (-\infty,s_\theta], \\[1mm] \sqrt{\|\al'\| - t^2}, \ & %t\in(s_\theta, s_\theta +\|\al'\|); \end{array} \right. 
%\end{equation}
and %the non-real spectrum of $\CA_\theta$ lies in the strip $\{z\in\C: |\Im\la|\le\|\al'\|\}$.
\begin{equation}
 \la\in \sigma(\CA_\theta) \ \implies \ |\Im \la | \le \|\alpha'\|, \ \ \Re \la \le s_\theta + \|\al'\|=:b_\theta.
\end{equation}
\end{theorem} 

\begin{proof}
All claims follow from Theorem \ref{B_theta} and Proposition \ref{stheta} by means of classical perturbation theorems for self-adjoint operators (see, e.g., \cite[Theorem~V.4.5]{MR0407617}) if we observe that $\CB_\theta=\CS_\theta + \CT_\theta$ and $\|\CT_\theta\| \le \|\alpha'\|$.
\end{proof}

The following corollary guarantees that $\CA_\theta$ has no eigenvalues in the closed right half-plane; it is an immediate consequence of  Theorem \ref{strip} and of the definition of $s_\theta$ in Proposition \ref{stheta}.

\begin{corollary}
\label{stable2}
Let $\theta\in[0,\infty]$. Then $\CA_\theta$ has no spectrum in the closed right half-plane~if 
\[
   %\|\al\| \, \displaystyle{\tan \bigg( \frac 12 \arctan \frac{2 \|\al\| \sqrt{\la_1(\theta)}} {\la_1(\infty) - \la_1(\theta)} \bigg)} 
   %+  \frac{\|\al'\|}{\sqrt{\la_1(\theta)}} < \sqrt{\la_1(\theta)}.
   \|\al\| < \sqrt{\la_1(\infty)}, \quad 
   \|\al'\| < \frac{\la_1(\infty)+\la_1(\theta)}2 - \sqrt{ \bigg( \frac{\la_1(\infty) - \la_1(\theta)}2\bigg)^2 \!\!\!+ \la_1(\theta) \|\al\|^2}.
\]
\end{corollary}

In a similar way as Proposition \ref{local}, the following local result can be proved.

\begin{proposition}
\label{question}
Let $\theta\in [0,\infty]$ and let $\lambda_0\in\sigma(\CS_\theta)$ be an eigenvalue of $\CS_\theta$ with multiplicity $m_0$. Set $\delta_0:= \dist\big( \la_0, \sigma(\CS_\theta) \setminus \{\la_0\} \big)/2$
and denote by $\Gamma_0$ the circle around $\lambda_0$  with radius $\delta_0$. If $\|\alpha'\|<\delta_0$, 
then the operator $\CA_\theta$ has $m_0$ eigenvalues inside~$\Gamma_0$ $($counted with multiplicities$)$; if $m_0=1$, then the unique eigenvalue of $\CA_\theta$ in $\Gamma_0$ is real.
%exactly one eigenvalue within, this eigenvalue is real and simple, and $\Gamma_0\subset\rho(\CA_\theta)$.
\end{proposition}

While in Section \ref{section4} the spectrum of the unperturbed operator $\CQ_\theta$ is known, this is not true for the unperturbed operator $\CS_\theta$. 
Therefore, in the remainder of this section, we investigate the spectrum of the semi-bounded self-adjoint operator $\CS_\theta$ by means of variational principles.

The quadratic form version of the min-max characterization of the eigenvalues of self-adjoint operators bounded from below 
(see \cite[Theorem 1, Section 2.2]{MR0477971}, \cite[Theorem~XIII.2]{MR0493421}) shows that
the eigenvalues $\la_k(\CS_\theta)$, $k\in\N$, of $\CS_\theta$ enumerated such that $\la_1(\CS_\theta) \ge \la_2(\CS_\theta) \ge \cdots$ and counted with multiplicities, 
are given~by
\begin{equation}
\label{var}
  \lambda_j(\CS_\theta)=\max_{\efrac{\CL\subset\CQ(\CS_\theta)}{\dim \CL=j}} \, \min_{\efrac{y\in \CL}{\|y\|=1}} \ {\mathfrak s}_\theta[y], \quad j=1,2,\dots.
\end{equation}
Here $\CQ(\CS_\theta)= \CD(A_\theta^{1/2}) \oplus \CD(A_\infty^{1/2})$ is the form domain of $\CS_\theta$ 
and ${\mathfrak s}_\theta$ is the corresponding quadratic form given by
\begin{equation}
\label{form}
  {\mathfrak s}_\theta[y] \!:=\!-\big(A_\theta^{1/2} y_1, A_\theta^{1/2} y_1\big) \!+\! \big(\al y_2, A_\theta^{1/2} y_1\big) 
                              \!+\!\big(A_\theta^{1/2} y_1, \al y_2\big) \!-\!\big(A_\infty^{1/2} y_2, A_\infty^{1/2} y_2\big)\!\!
\end{equation}
for $y=(y_1,y_2)^{\rm t} \in \CQ(\CS_\theta)$ (see \eqref{B2} and \cite[Section~VIII.6]{MR751959}). 
%One of the main applications of min-max principles is to compare eigenvalues of operators. Here we obtain the following comparison of 
%the (number of) positive eigenvalues of $\CS_\theta$.

\begin{proposition}
\label{july6}
The number $k_0$ of positive eigenvalues of the operator $\CS_\theta$ in $L_2(0,1)\oplus L_2(0,1)$ coincides with the number of positive eigenvalues of 
the shifted Bessel differential operator $-A_\infty + \alpha^2$ in $L_2(0,1)$, and we have
\begin{equation}
\label{vp}
   0<\lambda_j(\CS_\theta) \le \lambda_j(-A_\infty + \alpha^2),\quad j=1,2,\dots,k_0.
\end{equation}
\end{proposition}

\begin{proof}
Denote by $\kappa_0$ the number of positive eigenvalues of $-A_\infty+\alpha^2$. 
The inequality $k_0\le \kappa_0$ and the inequalities \eqref{vp} are immediate consequences 
of the right estimate in \eqref{ZB}, which implies that
\[
   {\mathfrak s}_\theta[y] = (\CS_\theta y,y) \le \big( (-A_\infty \!+ \al^2)y_2,y_2 \big)
\]
for $y=(y_1,y_2)^{\rm t} \!\in \CD(\CS_\theta)=\CD(A_\theta)\oplus \CD(A_\infty)$ and of the variational principle \eqref{var}. 

It remains to be shown that $k_0 \ge \kappa_0$. By the definition of $\kappa_0$ and by the variational principle 
\eqref{var} applied to the semi-bounded operator $-A_\infty+\al^2$, there exists a subspace 
$\CL \subset \CD(A_\infty^{1/2}) = \CQ(A_\infty)$ (the form domain of the positive operator~$A_\infty$) such that $\dim \CL = \kappa_0$ and
$$
  \min_{\efrac{x\in\CL}{\|x\|=1}} \big( - \big(A_\infty^{1/2}x, A_\infty^{1/2}x\big) + (\al x, \al x) \big) >0.
$$
Using \eqref{form}, it is easy to see that the subspace 
$
\left\{ \big( A_\theta^{-1/2}\!\alpha\,y_2, y_2 \big)^{\rm t} : y_2 \in \CL \right\} \subset \CQ(\CS_\theta)
$
has dimension $\kappa_0$ and
\begin{equation}
\label{shcy}
  {\mathfrak s}_\theta \biggl[ \binom{A_\theta^{-1/2}\alpha y_2}{y_2}\biggr]
  = \left( (-A_\infty+\al^2)y_2,y_2\right)>0, \quad y_2\in\CL\subset\CD(A_\infty^{1/2}).
\end{equation}
This and the variational principle \eqref{var} show that $\la_{\kappa_0}(\CS_\theta)>0$ and hence $\CS_\theta$ has at least $\kappa_0$ positive eigenvalues.  
\end{proof}

The eigenvalues of the block operator matrix $\CS_\theta$  coincide with the eigenvalues of its Schur complements 
\begin{alignat*}{2}
   S_{1,\theta}(\lambda)&:=-A_\theta -\lambda + A_\theta^{1/2}\alpha(A_\infty + \lambda)^{-1}\alpha A_\theta^{1/2}, \quad && \CD(S_{1,\theta}(\lambda))=\CD(A_\theta), \\
   S_{2,\theta}(\lambda)&:=-A_\infty \!-\lambda + \alpha A_\theta^{1/2}(A_\theta + \lambda)^{-1} A_\theta^{1/2}\alpha, \quad && \CD(S_{2,\theta}(\lambda))=\CD(A_\infty),
\end{alignat*}
which are defined for $\la \in \C \setminus \sigma(-A_\infty)$ and $\la \in \C \setminus \sigma(-A_\theta)$, respectively: 
\[
 \sigma_{\rm p}(\CS_\theta) \setminus \sigma(-A_\infty) = \sigma_{\rm p} (  S_{1,\theta} ), \quad
 \sigma_{\rm p}(\CS_\theta) \setminus \sigma(-A_\theta) = \sigma_{\rm p} (  S_{2,\theta} )
\]
(compare Proposition \ref{F_theta}). This allows us to characterize and estimate the eigenvalues of $\CS_\theta$ 
in the intervals $\big(-\!\lambda_1(\infty),\infty\big)$ and $\big(-\!\lambda_1(\theta),\infty\big)$ 
by variational principles for $S_{1,\theta}$ and $S_{2,\theta}$, respectively (see \cite{MR2068432} and, e.g., \cite[Section~2.10]{book}). 
As an example, we consider the eigenvalues of $\CS_\theta$ in $(-\la_1(\theta),\infty)$.
%This allows us to consider also negative eigenvalues and, at the same time, we obtain an alternative proof of Proposition \ref{july6}.

\begin{lemma}
\label{schur-est}
Let $\theta\in[0,\infty]$. The Schur complement $S_{2,\theta}$ satisfies the estimates
\begin{alignat*}{2} 
  S_{2,\theta}(\la) & \le - A_\infty + \frac{\la_1(\theta)}{\la_1(\theta)+\la} \, \al^2 - \la, \quad & &\la \in \big(-\la_1(\theta),0\big], \\
  S_{2,\theta}(\la) & \le - A_\infty +\al^2 - \la, \quad & &\la \in [0,\infty).
\end{alignat*}
Moreover, the derivative of $S_{2,\theta}$ is strictly negative and satisfies
\[
  S_{2,\theta}'(\la) \le - I, \quad  \la \in \big(-\la_1(\theta),\infty \big).
\]
\end{lemma}

\begin{proof}
Let $\la \in \big(-\la_1(\theta),\infty\big)$. The Schur complement $S_{2,\theta}$ can be rewritten as
\begin{align*}
  S_{2,\theta}(\la) %&= -A_\infty - \la + \al^2 + \al \big( A_\theta (A_\theta + \la)^{-1} - I \big) \al \\
                    &= -A_\infty - \la + \al A_\theta (A_\theta + \la)^{-1} \al \\
                    &= -A_\infty - \la + \al^2 - \la \, \al (A_\theta + \la)^{-1} \al.
\end{align*}
Since $A_\theta \ge \la_1(\theta) > 0$ (see Proposition \ref{A_theta}), the resolvent of $A_\theta$ satisfies the 
two-sided inequality $0\le (A_\theta + \la)^{-1} \le 1/(\la_1(\theta)+\la)$. 
For $\la \in \big(-\la_1(\theta),0\big]$, the right inequality yields the first estimate claimed for $S_{2,\theta}(\la)$, while for $\la \in [0,\infty)$ the
left inequality yields the second one.

The inequality for the derivative of $S_{2,\theta}(\la)$ follows from the identity
\[
  S_{2,\theta}'(\lambda)= - I - \alpha A_\theta^{1/2}(A_\theta + \lambda)^{-2} A_\theta^{1/2}\alpha, \quad  \la \in \big(-\la_1(\theta),\infty \big).
\vspace{-6mm}
\]
\end{proof}

\vspace{2mm}

\begin{proposition}
\label{july11}
Let $\theta \in[0,\infty]$. For every subinterval $[a,\infty) \subset \big(-\!\lambda_1(\theta),\infty\big)$ let
$\la_{k_a}(\CS_\theta) \le \dots \le \la_1(\CS_\theta)$ be the eigenvalues of $\,\CS_\theta$ in $[a,\infty)$ $($counted with multiplicities$)$. 
Then 
\begin{equation}
\label{schur-var}
 \la_j(\CS_\theta) = \min_{\efrac{\CL\subset \CD(A_\infty)}{\dim \CL=j}} \, \max_{\efrac{y_2\in\CL}{\|y_2\|=1}} \ p_{2,\theta}(y_2), \quad j=1,2,\dots, k_a,
\end{equation}
where $p_{2,\theta}(y_2)$ is the $($unique$)$ zero of $\,(S_{2,\theta}(\cdot)y_2,y_2)$ on $[a,\infty)$ if a zero exists and $p_{2,\theta}(y_2):=-\infty$ otherwise for $y_2\in \CD(A_\infty)$. Moreover, the number $k_a$ of eigenvalues of $\CS_\theta$ in $[a,\infty)$ is given by
\[
 k_a = \dim \CL_{[0,\infty)}\big(S_{2,\theta}(a)\big);
\]
here $\CL_I(S_{2,\theta}(\la))$ denotes the spectral subspace of the self-adjoint operator $S_{2,\theta}(\la)$ corresponding to an interval $I \subset \big( -\la_1(\theta),\infty \big)$.
\end{proposition}

\begin{proof}
The Schur complement $S_{2,\theta}$ satisfies the assumptions of \cite[Theorem~2.1]{MR2068432} on every subinterval of $\big(-\!\lambda_1(\theta),\infty\big)$: $\CD(S_{2,\theta}(\la)) = \CD(A_\infty)$ is independent of $\la$, $S_{2,\theta}$ is strictly decreasing with $S_{2,\theta}'\le -I$, and $(S_{2,\theta}(\la)y_2,y_2) \to -\infty$,  $\la\to\infty$, for all $y_2 \in \CD(A_\infty)$ by Lemma \ref{schur-est}. 
Now all claims follow from the fact that the eigenvalues of $\CS_\theta$ in $\big( -\la_1(\theta),\infty \big)$ coincide with those of $S_{2,\theta}$
by Proposition \ref{F_theta} and from the variational principle in \cite[Theorem~2.1]{MR2068432} applied to $S_{2,\theta}$.
% In order to formulate this variational principle, e.g., in the first case, we observe that  for $x\in L_2(0,1),\,x\ne 0$,  the function $\lambda\mapsto (S_{1,\theta}(\lambda)x,x))$ is strictly decreasing in the interval $(-\lambda_1(\infty),\infty)$. Therefore it has at most one zero in this interval, which we denote  by $p(x)$; if it has no zero in this interval we set $p(x)=-\infty$. Then we obtain for the eigenvalues of $\CS_\theta$ which are $>-\lambda_1(\infty)$, say $\lambda_1(\CS_\theta)\ge \lambda_2(\CS_\theta)\ge\cdots\ge\lambda_\ell(\CS_\theta)$:
% $$
% \lambda_j(\CS_\theta)=\max_{\CL:\,\dim \CL=j}\, \min_{x\in \CL,x\ne 0} p(x),\quad j+1,2,\dots,\ell.
% $$
\end{proof}

\begin{remark}
Setting $a=0$ in Proposition \ref{july11}, we obtain another proof of Proposition \ref{july6}. In fact, 
the estimate \eqref{vp} follows from the variational principle \eqref{schur-var} if we observe that $S_{2,\theta}(\la) \le -A_\infty +\al^2 - \la$ for $\la\in[0,\infty)$ by Lemma \ref{schur-est} and hence
\[
  p_{2,\theta}(y_2) \le \frac {\big( (-A_\infty+\al^2)y_2,y_2\big)}{(y_2,y_2)}, \quad y_2 \in \CD(A_\infty),
\]
the right hand side being the zero of the function $\la \mapsto \big( (-A_\infty +\al^2 - \la )y_2,y_2 \big)$.
Further, since $S_{2,\theta}(0)=-A_\infty + \alpha^2$, Proposition \ref{july11} shows that the number $k_0$ of positive eigenvalues of $\CS_\theta$ is given by
\[
 k_0=\dim \CL_{[0,\infty)} \big( S_{2,\theta}(0) \big) = \dim \CL_{[0,\infty)} (-A_\infty + \alpha^2).
\]
\end{remark}

\section{Comparison of the two eigenvalue estimates}
\label{section5}

The two eigenvalue estimates obtained in Theorem \ref{!!!} and Theorem~\ref{strip} show, in particular, that every eigenvalue of the 
dynamo operator $\CA_\theta$ satisfies the inequalities
\[
 |\Im \la| \le \|\al'\|, \quad \Re \la \le \min\{ a_\theta, b_\theta \};
\]
here the uniform bound for the imaginary parts was proved in Theorem \ref{strip} and $a_\theta$ and $b_\theta=s_\theta+\|\al'\|$ are the right bounds for $\sigma(\CA_\theta)$ derived in Theorems~{\rm \ref{!!!}} and~{\rm \ref{strip}}, respectively. In this section we prove that, apart from a small bounded set of values of $\|\al\|$ and $\|\al'\|$, 
we always have $a_\theta < b_\theta$. Hence, in general, a combination of the two eigenvalue estimates from Theorems~{\rm \ref{!!!}} and~{\rm \ref{strip}} yields the best result.

First we present two auxiliary technical lemmas. The graphs of the functions introduced therein are displayed in Figure \ref{k-i} below.

\begin{lemma}
\label{very-hot}
The functions 
\begin{alignat*}{2}
  & k_1\!:\! [0,\!\sqrt{\la_1(\infty)}] \!\to\! [0,\infty), \ \ k_1(t) :=&&\!\! \frac{\la_1(\infty)\!+\!\la_1(\theta)}2 - \sqrt{ \!\bigg( \frac{\la_1(\infty) - \la_1(\theta)}2\bigg)^{\!\!2} \!\!+ \la_1(\theta)\,t^2}, \\ 
  & k_2\!:\!(0,\!\sqrt{\la_1(\infty)}] \!\to\! [0,\infty), \quad && k_2(t):= \frac{\la_1(\infty)\!-t^2}{t} \sqrt{\la_1(\theta)}, \\
  & k_3\!:\!(0,\!\sqrt{\la_1(\infty)+\la_1(\theta)}] \!\to\! [0,\infty), \quad &&  k_3(t):= \frac{\la_1(\infty)\!-\!t^2 + \la_1(\theta)}{t} 2 \sqrt{\la_1(\theta)},
\end{alignat*}
are continuous and strictly decreasing with
\begin{align*}
k_1(\sqrt{\la_1(\infty)})=k_2(\sqrt{\la_1(\infty)})=k_3(\sqrt{\la_1(\infty)+\la_1(\theta)})=0.
\end{align*}
They satisfy the inequalities
\begin{equation}
\label{k_i}
\begin{array}{l}
  k_1 < k_2  \ \text{ on } \ (0,\sqrt{\la_1(\infty)}), \\[1mm]
  k_2 < k_3 \ \text{ on } \ (0,\!\sqrt{\la_1(\infty)+\la_1(\theta)}),
\end{array}
\end{equation}
where we have set $k_2(t):=0$ for $t\in (\sqrt{\la_1(\infty)},\sqrt{\la_1(\infty)+\la_1(\theta)}]$.
Moreover, 
\begin{alignat*}{3}
  &b_\theta \le 0 && \ \iff \ \|\al\| \le \sqrt{\la_1(\infty)}, \ \ && \|\al'\| \le k_1(\|\al\|),\\
  &a_\theta \le 0 && \ \iff \ \|\al\| \le \sqrt{\la_1(\infty)}, \ \ && \|\al'\| \le k_2(\|\al\|),\\
  0 < \,& a_\theta \le \la_1(\theta) && \ \iff \ \|\al\| \le \sqrt{\la_1(\infty)+\la_1(\theta)}, \ \ && \|\al'\| \le k_3(\|\al\|).
\end{alignat*}
\end{lemma}

\begin{proof}
The claims for the functions $k_1$, $k_2$, and $k_3$ are easy to check. The last two equivalences follow if we use the formulas for $a_\theta$ in Theorem \ref{!!!} and solve the corresponding inequalities on the left hand sides for $\|\al'\|$. The condition $b_\theta = s_\theta + \|\al'\| \le 0$ is satisfied if and only if $s_\theta\le 0$ and $\|\al'\|<-s_\theta$. It remains to use the respective formula for $s_\theta$ from Proposition \ref{stheta} to obtain the first~equivalence. 
\end{proof}

\begin{lemma}
\label{aux}
Define two functions $k_4^\pm$ implicitly by the equation
\[
  \sqrt{\!\Bigl( \frac{\la_1(\infty)\!-\!\la_1(\theta)}2 \Bigr)^{\!2} \!\!\!+ \!\la_1(\theta) t^2} + k_4^\pm(t) =
  \frac{t^2\!}2 + \sqrt{\!\Bigl( \frac{\la_1(\infty)\!-\!\la_1(\theta)\!-\!t^2\!}2 \Bigr)^{\!2} \!\!\!+ \!\sqrt{\la_1(\theta)} \, t \,k_4^\pm(t)}
\]
for $t\in [0,\sqrt{\la_1(\infty)}]$ and let 
\begin{align*}
  &k_5: \left[\!\sqrt{\la_1(\infty)}, \dfrac{\sqrt{\la_1(\theta)}}2 + \sqrt{\!\la_1(\infty)-\dfrac{3\la_1(\theta)}4} \right] \to \big[0,\infty\big), \\ 
  &k_5(t) := \la_1(\infty)-\la_1(\theta) + t \sqrt{\la_1(\theta)} - t^2. 
\end{align*}
Then the graphs of $\,k_4^-\!$, $k_4^+\!$, and $\,k_5$ form a continuous curve $\Gamma^{\rm ex}$ connecting the~points 
\begin{align*}
 C_1 &:= \Big(\sqrt{\la_1(\infty)},0\Big), \quad C_2:=\Big(\sqrt{\la_1(\infty)},\sqrt{\la_1(\infty)\la_1(\theta)}-\la_1(\theta)\Big),\\
 C_3 &:=\Big( \!\frac{\sqrt{\la_1(\theta)}}2\!+\!\sqrt{\la_1(\infty)\!-\!\dfrac{3\la_1(\theta)}4},0 \Big).
\end{align*}
Let $\Delta^{\rm ex}$ be the open bounded set surrounded by the curve $\Gamma^{\rm ex}$ and the segment $\overline{C_1C_3}$ on the ordinate axis. Then 
\begin{align*}
 \Delta^{\rm ex} \subset \big\{ \big(\|\al\|,\|\al'\|\big) : \, k_2(\|\al\|) < \|\al'\| < k_3(\|\al\|), \ \|\al'\| < \sqrt{\la_1(\infty)\la_1(\theta)}-\la_1(\theta) \big\}
\end{align*}
and, for $0 < a_\theta \le \la_1(\theta)$, we have 
\begin{align*}
 b_\theta = a_\theta \ &\iff \ \big(\|\al\|,\|\al'\|\big) \in \Gamma^{\rm ex}, \\
 b_\theta < a_\theta \ &\iff \ \big(\|\al\|,\|\al'\|\big) \in \Delta^{\rm ex}.
\end{align*}
\end{lemma}

\begin{proof}
It is easy to see that the graph of the function $k_5$, which lies to the right of the vertical line $\|\al\|=\sqrt{\la_1(\infty)}$, 
is strictly decreasing from the point $C_2$ on this line to $C_3$, the zero of $k_5$. Moreover, if $0 < a_\theta \le \la_1(\theta)$ and 
$\|\al\| \ge \sqrt{\la_1(\infty)}$, the formulas for $a_\theta$ in Theorem \ref{!!!} (ii) and $b_\theta=s_\theta+\|\al'\|=-\la_1(\infty)+\|\al\|^2+\|\al'\|$ (see Proposition~\ref{stheta})
show that $k_5 \ge 0$ on the interval where it is defined,
\[
    b_\theta \le a_\theta \ \iff \ %\sqrt{\la_1(\cl{\infty})} \le \|\al\| \le \cl{\dfrac{\sqrt{\la_1(\theta)}}2 + \sqrt{\!\la_1(\infty)-\dfrac{3\la_1(\theta)}4}}, \ \ 
                                 \|\al'\| \le k_5(\|\al\|),
\]
and that equality holds on the graph of $k_5$.

By elementary calculations (see Remark \ref{august20} below), one can show that the graphs of the implicitly defined functions $k_4^-$ and $k_4^+$ form an arc from $C_1$ to $C_2$
lying to the left of the vertical line $\|\al\|=\sqrt{\la_1(\infty)}$. Furthermore, 
if $0 < a_\theta \le \la_1(\theta)$ and $\|\al\| \le \sqrt{\la_1(\infty)}$, the formulas for $a_\theta$ in Theorem \ref{!!!} (ii)
and for $b_\theta=s_\theta+\|\al'\|$ with $s_\theta$ as in Proposition \ref{stheta} show that
\[
    b_\theta \le a_\theta \ \iff \ k_4^-(\|\al\|) \le  \|\al'\| \le k_4^+(\|\al\|)
\]
and that equality holds on the graphs of $k_4^\pm$.
\end{proof}

\begin{remark}
\label{august20}
The two functions $k_4^\pm$ in Lemma \ref{aux} can be calculated explicitly:
\begin{align*}
 &k_4^\pm(t)= \frac{t^2}2 + \frac t2 \sqrt{\la_1(\theta)} - \sqrt{ \Bigl( \frac{\la_1(\infty)-\la_1(\theta)}2 \Bigr)^2 \!\! + \la_1(\theta) t^2 }\\
 &             \pm \!\sqrt{ \!\Bigl( \frac{\la_1(\infty)\!-\!\la_1(\theta)\!-\!t^2\!}2 \Bigr)^{\!2} \!\!\!+ \frac{t^2\!}4 \la_1(\theta) - t \sqrt{\la_1(\theta)} 
                            \Bigl( \sqrt{\!\Bigl( \frac{\la_1(\infty)\!-\!\la_1(\theta)}2 \Bigr)^{\!2} \!\!\!+ \la_1(\theta) t^2 } - \frac{t^2}2 \!\Bigr)\! } 
\end{align*}
for $t\in [0,\sqrt{\la_1(\infty)}]$ such that the term under the last square root is positive. The set of these $t$ is a (small) interval of the form $[\mu,\sqrt{\la_1(\infty)}]$. 
A formula for $\mu$ may be found, e.g., with MAPLE, but it is extremely involved; a lower bound for $\mu$ is the point where $k_2$
attains the value $\sqrt{\la_1(\theta)\la_1(\infty)}-\la_1(\theta)$ (the height of~$\Delta^{\rm ex}$):
\begin{equation}
\label{mu}
  \mu \ge- \dfrac{\sqrt{\la_1(\infty)} \!+\! \sqrt{\la_1(\theta)}}2+\sqrt{\Bigl( \dfrac{\sqrt{\la_1(\infty)} \!+\! \sqrt{\la_1(\theta)}}2\Bigr)^2 \!\!\!+\la_1(\infty)}. 
\end{equation}
\end{remark}

% For the imaginary parts of the eigenvalues, Theorem~\ref{strip} yields the uniform bound $\|\al'\|$, while the bound from Theorem~\ref{!!!} 
% tends to $\infty$ in the left half-plane.
% For the real part of the right-most eigenvalue, the situation is opposite,
% here the first bound $a_\theta$ tends to be sharper. It may even happen that the second bound $b_\theta=s_\theta+\|\al'\|$ is positive, whereas $a_\theta$ is negative.
% There is only a small bounded set of values of $\|\alpha\|$ and $\|\alpha'\|$ for which $b_\theta < a_\theta$ and the second estimate is globally better than the first one.

% \begin{proposition}
% \label{compare1}
% Let $\theta \in [0,\infty]$. Then there exists a point $\la_\theta \le \min\{ a_\theta, b_\theta\}$ with 
% \[
%   h_\theta (\Re \la) > \|\al'\|, \quad \Re\la < \la_\theta,
% \]
% that is, the bound for the non-real spectrum in Theorem {\rm \ref{strip}} is better than the one from Theorem {\rm \ref{!!!}} to the left of $\la_0$.
% \end{proposition}
% 
% \begin{proof}
% The claim is immediate from the fact, by Theorem \ref{!!!}, that the function $h_\theta$ is strictly decreasing, tends to $\infty$ for $\Re \la \to - \infty$ and $h(a_\theta) = 0$.
% \end{proof}

\begin{proposition}
\label{compare2}
Let $\theta\in[0,\infty]$. Then the right bounds $a_\theta$ and $b_\theta$ for $\sigma(\CA_\theta)$ established 
in Theorems {\rm \ref{!!!}} and {\rm \ref{strip}}, respectively, satisfy
\begin{align*}
 a_\theta < b_\theta \ &\iff \  \big(\|\al\|,\|\al'\|\big) \in \big( [0,\infty) \times [0,\infty) \big) \setminus \Delta^{\rm ex}, \\
 a_\theta = b_\theta \ &\iff \ \big(\|\al\|,\|\al'\|\big) \in \Gamma^{\rm ex},
\end{align*}
where the bounded set $\Delta^{\rm ex}\!$ and the curve $\Gamma^{\rm ex}\!$ are as in Lemma {\rm \ref{aux}}. Moreover, 
$$a_\theta < 0 < b_\theta \ \iff \ \|\al\|^2 < \la_1(\infty), \  \ k_1(\|\al\|) < \|\al'\| <  k_2(\|\al\|);$$
the last condition can be written equivalently as
\[
   \frac{\la_1(\infty)\!+\!\la_1(\theta)}2 - \sqrt{ \bigg( \frac{\la_1(\infty) - \la_1(\theta)}2\bigg)^{\!2} \!\!\!+ \la_1(\theta) \|\al\|^2} < \|\al'\| < \frac{\la_1(\infty)\!-\!\|\al\|^2}{\|\al\|} \sqrt{\la_1(\theta)}.
\]
%\begin{hlist}
% \item[{\rm (i)}] $\,0< \|\alpha\|^2 + \dfrac{\|\al\|\,\|\alpha'\|}{\sqrt{\lambda_1(\theta)}} \le \lambda_1(\infty) 
% \ \ \big(\!\Longleftrightarrow -\la_1(\theta) <  a_\theta \le 0 \big)${$:$} 
% \[
%    a_\theta < b_\theta.
% \]
% \item[{\rm (ii)}] $\,\|\alpha\|^2 \!+ \dfrac{\|\al\|\,\|\alpha'\|}{\sqrt{\lambda_1(\theta)}} \!>\! \lambda_1(\infty)$, $\|\alpha\|^2 \!+ \dfrac{\|\al\|\,\|\alpha'\|}{2\sqrt{\lambda_1(\theta)}} \!\le\! \lambda_1(\infty)+\la_1(\theta) \, \big(\!\!\Longleftrightarrow\!  0 \!<\!  a_\theta  \!\le\! \la_1(\theta)\big)${$:$}
% \item[{\rm (iii)}] $\,\|\alpha\|^2 + \dfrac{\|\al\|\,\|\alpha'\|}{2\sqrt{\lambda_1(\theta)}} > \lambda_1(\infty)+ \la_1(\theta) 
% \ \ \big(\!\Longleftrightarrow \la_1(\theta) < a_\theta \big)${$:$}
% \[
%    a_\theta < b_\theta.
% \]
%\end{hlist}
\end{proposition}

Figure \ref{k-i} illustrates the sets occurring in Proposition~\ref{compare2}:
the exceptional set $\Delta^{\rm ex}$ where $b_\theta < a_\theta$ is the highlighted small bounded set between the graphs of $k_2$ and $k_3$;
the set where $a_\theta < 0 < b_\theta$ is the highlighted unbounded set enclosed by the graphs of $k_1$ and $k_2$. 

\begin{figure}[h]
%\vspace{1mm}
\parbox{7cm}
{\epsfysize=7cm
\epsfbox{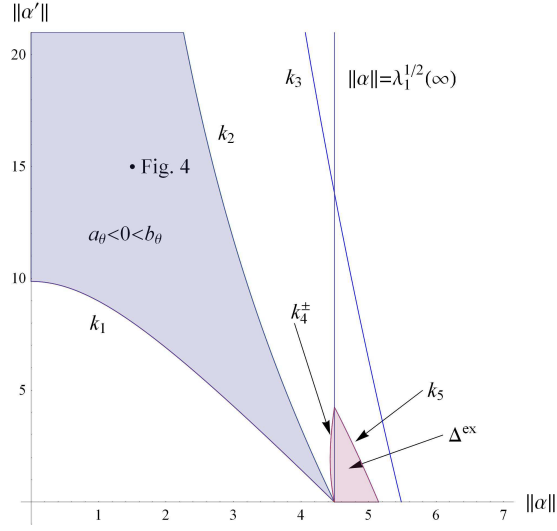}}
\caption{\label{k-i} The graphs of the functions $k_1$, $k_2$, $k_3$, $k_4^\pm$, and $k_5$.}
%\vspace*{-1mm}
\end{figure} 

\begin{remark} 
Proposition \ref{compare2} shows, in particular, that we always have $a_\theta < b_\theta$ if one of the following \vspace{2mm} holds:
\begin{hlist}
 \item[i)] $\|\al'\|> \sqrt{\la_1(\infty)\la_1(\theta)}-\la_1(\theta)$,\vspace{1mm}
 %\item[ii)] $\|\al\| < \dfrac{\sqrt{\la_1(\infty)} \!-\! \sqrt{\la_1(\theta)}}2+\sqrt{\Bigl( \dfrac{\sqrt{\la_1(\infty)} \!-\! \sqrt{\la_1(\theta)}}2\Bigr)^2 \!\!\!+\la_1(\infty)}$, %wrong!?!
 \item[ii)] $\|\al\| < - \dfrac{\sqrt{\la_1(\infty)} \!+\! \sqrt{\la_1(\theta)}}2+\sqrt{\Bigl( \dfrac{\sqrt{\la_1(\infty)} \!+\! \sqrt{\la_1(\theta)}}2\Bigr)^2 \!\!\!+\la_1(\infty)}$,
 \item[iii)] $\|\al\| >  \dfrac{\sqrt{\la_1(\theta)}}2 + \sqrt{\la_1(\infty)-\dfrac{3\la_1(\theta)}4}$.%\sqrt{\la_1(\infty)\!+\!\la_1(\theta)}$. 
\vspace{2mm}
\end{hlist}
% The exceptional set $\Delta^{\rm ex}$ is contained, in particular, in the rectangle given by the inequalities
% \begin{align*}
%   &\frac{\sqrt{\la_1(\infty)} \!-\! \sqrt{\la_1(\theta)}}2+\sqrt{\Bigl( \frac{\sqrt{\la_1(\infty)} \!-\! \sqrt{\la_1(\theta)}}2\Bigr)^2 \!\!\!+\la_1(\infty)} \le \|\alpha\|\le \sqrt{\la_1(\infty)\!+\!\la_1(\theta)}, \\
%   &0<\|\al'\| \le \sqrt{\la_1(\theta)\la_1(\infty)}-\la_1(\theta). 
% \end{align*}
This follows from Lemma~\ref{aux} if we observe that the lower bound in i) is the maximal value of $\|\al'\|$ in $\Delta^{\rm ex}$, the upper bound in ii) is the lower bound for the left end-point
$\mu$ of the domain of definition of the functions $k_4^\pm$ (see Remark \ref{august20}), and  the lower bound in iii) is the right end-point of the domain of definition of the function $k_5$.
%where $k_5$ becomes $0$.
\end{remark}

%\begin{proof}
{\bf Proof of Proposition \ref{compare2}.}
In order to compare $a_\theta$ and $b_\theta$, we have to distinguish the cases (i), (ii), and (iii)  in Theorem \ref{!!!} for $a_\theta$  
and the two cases $\|\al\|\le \sqrt{\la_1(\infty)}$ and $\|\al\| > \sqrt{\la_1(\infty)}$ for $b_\theta=s_\theta+\|\al'\|$ according to the definition of $s_\theta$ in Proposition \ref{stheta}.

Elementary but lengthy and tedious calculations show that in the cases $a_\theta \le 0$ (Theorem \ref{!!!} (i)) and
$a_\theta > \la_1(\theta)$ (Theorem \ref{!!!} (iii)), the equation $a_\theta=b_\theta$ has no solution 
and $a_\theta < b_\theta$. Moreover, Lemma \ref{very-hot} implies that the cases $a_\theta<0$ and $b_\theta>0$ appear simultaneously 
if and only if $\|\al\|< \sqrt{\la_1(\infty)}$ and  $k_1(\|\al\|) < \|\al'\| <  k_2(\|\al\|)$ since we have $k_1 < k_2$ by \eqref{k_i}.

In the case $0< a_\theta \le \la_1(\theta)$ (Theorem \ref{!!!} (ii)), Lemma \ref{aux} shows that
$a_\theta=b_\theta$ on the curve $\Gamma^{\rm ex}$ and $a_\theta < b_\theta$ if and only if $(\|\al\|,\|\al'\|) \notin \Delta^{\rm ex}$.  \qed

\vspace{2mm}

\begin{figure}[h]
\vspace{1mm}
\parbox{6cm}
{\epsfysize=6cm
\epsfbox{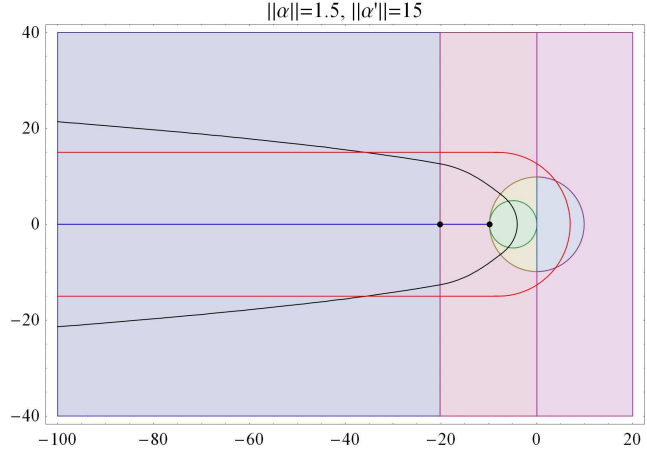}}
\caption{\label{stab} Spectral enclosures of Theorems \ref{!!!} and \ref{strip} ($a_\theta \!<\! 0 \!<\! b_\theta$).}
\end{figure} 

In Figure \ref{stab} the boundaries of the two spectral enclosures from Theorem \ref{!!!} and Theorem \ref{strip} are displayed together; 
the set $\Sigma_\theta$ from Theorem \ref{!!!}  is the set with unbounded imaginary part and smaller real part $\le a_\theta$, the enclosing set from Theorem~\ref{strip} 
is the set with bounded imaginary part and larger real part $\le b_\theta$. Here the values $\|\al\| = 1.5$ and $\|\al'\|=15$ are chosen such that $a_\theta < 0 < b_\theta$; the corresponding point 
$(\|\al\|,\|\al'\|)=(1.5,15)$ is marked in Figure \ref{k-i} by a black dot.
 
Proposition \ref{compare2} shows that, in general, the upper bound of Theorem \ref{!!!} improves the estimate of Theorem~\ref{strip}; 
only for values $(\|\al\|,\|\al'\|)\in \Delta^{\rm ex}$, the enclosure of Theorem~\ref{strip} is better. The following corollary is a direct consequence
of all these results.

\begin{corollary}
Let $\theta \in [0,\infty]$, let $a_\theta$, $b_\theta$ be the bounds established in Theorems~\ref{!!!} and \ref{strip}, respectively, and let $\Delta^{\rm ex} \subset [0,\infty) \times [0,\infty)$ be the
bounded set defined in~Lemma \ref{aux}. Then every eigenvalue $\la$ of the dynamo operator $\CA_\theta$  satisfies
\[
 | \Im \la | \le \|\al'\|, \quad \Re \la \le \left\{ \begin{array}{ll} a_\theta \ & \mbox{ if } \ \big(\|\al\|,\|\al'\|\big) \in \big([0,\infty) \times [0,\infty)\big) \setminus \Delta^{\rm ex}, \\[1mm] 
                                                                       b_\theta \ & \mbox{ if } \ \big(\|\al\|,\|\al'\|\big) \in \Delta^{\rm ex}. 
                                                     \end{array} \right.
\]
\end{corollary}

The case that the bound for the real part is less than $0$ is of particular interest from the physical point of view since the onset of the dynamo effect requires supercritical
modes, that is, eigenvalues with real part greater than $0$.

Proposition \ref{compare2} shows that, if $a_\theta<0$, then we always have $a_\theta< b_\theta$.
%Hence, in this case, the bound $a_\theta$ in Theorem \ref{!!!} yields a tighter bound for the rightmost eigenvalue of $\CA_\theta$ than Theorem \ref{strip}. 
In particular, this shows that Corollary \ref{stable} is stronger than Corollary \ref{stable2}; for the special case $\theta=l$ this gives:

\vspace{2mm}

\noindent
{\bf Anti-dynamo theorem.}
%Let $\theta \in [0,\infty]$. Then 
The dynamo operator $\CA_{l}$ has no spectrum in the closed right half-plane if 
\begin{equation}
\label{stable1} 
  \|\alpha\|^2 + \dfrac{\|\al\|\|\alpha'\|}{\sqrt{\lambda_1(l)}} < \lambda_1(\infty).
\end{equation}

%{\bf Non-oscillation theorem.} 
%The eigenvalues of the operator \cA_\theta remain purely real if ... ... ...

Note that, since the above result has been obtained by operator theoretic estimates,
% due to the specific operator-theoretic estimation techniques used for the derivation of the above result, 
a violation of the condition \eqref{stable1} 
%and (eq. No vom non-oscillation theorem) 
provides only a necessary, not a sufficient, condition for the existence of supercritical %and oscillatory 
dynamo~regimes.

\begin{remark}
By Lemma \ref{may}, we know that $\la_1(l)$ and $\la_1(\infty)$ are the first non-zero zeros of the Bessel functions $J_{l-1/2}(\sqrt{\la})$ and $J_{l+1/2}(\sqrt{\la})$, respectively.
If we indicate this dependence on the parameter $l\in \N$ for a moment and write $\la_1(l)=\la_1(l;l)$ and $\la_1(\infty)=\la_1(\infty,l)$, we see that
$\la_1(l_1;l_1) < \la_1(l_2;l_2)$ and $\la_1(\infty;l_1) < \la_1(\infty;l_2)$~for $l_1,l_2 \in \N$ with $l_1 < l_2$ (see \cite[9.5.2]{abramowitz+stegun}).
Therefore, if condition \eqref{stable1} is satisfied for $l_1\in\N$, then it is also satis\-fied for $l_2 \in \N$ with $l_2 > l_1$.
%Taking into account that  $\lambda_1(\theta, l_1)<\lambda_1(\theta,l_2)$ for $l_1<l_2$, 

This provides an explanation of the numerical observation that, for increasing $\|\alpha\|$ and $\|\alpha'\|$, criticality usually starts from dipole modes ($l=1$) and quadrupole and higher-degree modes ($l>1$) become supercritical only for larger values of $\|\alpha\|$ and~$\|\alpha'\|$ (see \cite{stefani-2003-67}).
\end{remark}

\section{Examples}
\label{section-last}

In this section we illustrate our results by applying them to the physical dynamo problem \eqref{diffsyst}, \eqref{boundcond} and to the idealized dynamo problem \eqref{diffsyst}, \eqref{boundcond-id}. 
Here we have to set $\theta=l$ and $\theta=\infty$, respectively, in the previous statements.

The simplest case of constant $\alpha$, where all eigenvalues are real, already shows that the physical problem is far more complex than the idealized problem: for the latter explicit formulas for the eigenvalues are known, whereas the eigenvalues of the former are given only implicitly as the zeros of an equation involving 4 different Bessel~functions. 

\begin{example}
\label{id-ex}
Let $\al\equiv \alpha_0$ be constant and consider the idealized dynamo problem \eqref{diffsyst}, \eqref{boundcond-id} for which $\theta=\infty$.
%
%\emph{Calculation of the eigenvalues.}
The eigenvalues of this problem, or equivalently of the operator $\CA_\infty$, can be calculated explicitly (see, e.g., \cite{MR2252699}).
In fact, since $\alpha$ is constant and hence $\alpha'\equiv 0$, the operator $\CA_\infty$ is given  by
\[
 \CA_\infty=\matrix{rc}{-A_\infty & \alpha_0 \\ \al_0 A_\infty & - A_\infty}, \quad \CD(\CA_\infty) = \CD(A_\infty) \oplus \CD(A_\infty).
\]
A point $\la\in\C$ is an eigenvalue of $\CA_\infty$ if and only if there is a $y=(y_1,y_2)^{\rm t} \in \CD(\CA_\infty)$, $y\ne 0$, such that
\begin{align*}
 (-A_\infty-\la) y_1 + \al_0 y_2 = 0,\\
 \al_0 A_\infty y_1 + (-A_\infty-\la) y_2 = 0.
\end{align*}
Since $y_2 \in \CD(A_\infty)$, the first equation yields $y_1 \in \CD(A_\infty^2)$ and thus 
the eigen\-value equations are equivalent to the relations
\begin{align*}
   \big( (A_\infty + \la)^2 -\al_0^2 A_\infty \big) y_1 = 0, \quad \al_0 y_2 = (A_\infty+\la) y_1.
\end{align*}
Hence $\la\in \sigma_{\rm p}(\CA_\infty)$ if and only if $0\in \sigma_{\rm p}\big((A_\infty + \la)^2 -\al_0^2 A_\infty\big)$. 
Now the spectral mapping theorem shows that %$\sigma_{\rm p}(\CA_\infty) = \{  \la_n^\pm(\CA_\infty) : n\in\N \}$ with 
\[
  \sigma(\CA_\infty) = \big\{ \!-\!\la_n(\infty) \pm \al_0 \sqrt{\la_n(\infty)}: n\in \N \big\},
\]
where $\la_n(\infty)$, $n\in\N$, are the eigenvalues of the Bessel operator $A_\infty$ with~Dirichlet boundary conditions
(that is, $\la_n(\infty)$ is the $n$-th zero of $\la \mapsto J_{l+1/2}(\sqrt{\la})$; see Definition \ref{A_theta-def} and Remark \ref{may} i)).
Thus we have
\begin{align}
\label{eig-val-id}
  \max \sigma(\CA_\infty) &= -\la_1(\infty) + |\al_0| \sqrt{\la_1(\infty)} \\
\intertext{and the estimate}
\label{eig-val-est}
  \max \sigma(\CA_\infty) &= - \left( \sqrt{\la_1(\infty)} - \frac{|\al_0|}2 \right)^2 + \frac{\al_0^2}4 \le \frac{\al_0^2}4.
\end{align}

%\emph{Application of abstract results}.
The abstract eigenvalue estimates in Theorem \ref{!!!} and in Theorem \ref{strip}, combined with Remark \ref{siginf} iii), 
show that, for $\theta=\infty$ and $\|\al'\|=0$,
\begin{equation}
\label{eig-val-id-est}
  \max \sigma(\CA_\infty) %= \max \sigma_{\rm p}(\CS_\theta) 
  \le a_\infty=b_\infty = - \la_1(\infty) + |\al_0| \max \big\{ |\al_0|, \sqrt{\la_1(\infty)}\,\big\};
\end{equation}
note that $\CB_\infty=\CS_\infty$ and thus $b_\infty=s_\infty$ because $\alpha$ is constant (see Proposition \ref{B_theta-sim} and Theorem \ref{B_theta}). 

Comparing \eqref{eig-val-id} and \eqref{eig-val-id-est}, we see that the abstract upper bounds $a_\infty$ and $b_\infty$ for $\sigma(\CA_\infty)$ 
in Theorems \ref{!!!} and \ref{strip} are sharp for the case $\|\al\| \le \sqrt{\la_1(\infty)}$.
Moreover, \eqref{eig-val-est} coincides with the abstract estimate proved in Proposition \ref{alpha-const} for constant~$\alpha$.
\end{example}

For physical boundary conditions, even in the case of constant $\alpha$, only implicit formulas are known for the eigenvalues 
$\la_n(\CA_l)$ of $\CA_l$ (compare \cite[Section 14, in particular, (14.39), (14.41), (14.42)]{MR668520}). 

\begin{example}
\label{phys-ex} 
Let $\al\equiv \alpha_0$ be constant and consider the physical dynamo problem \eqref{diffsyst}, \eqref{boundcond} for which $\theta=l$. 
There are no explicit formulas known for the eigenvalues of this problem, or equivalently of the operator $\CA_l$ given by
\[
 \CA_l=\matrix{rc}{-A_l & \alpha_0 \\ \al_0 A_l & - A_\infty}, \quad \CD(\CA_l) = \CD(A_l) \oplus \CD(A_\infty).
\]
Note that here the operators $A_l$ and $A_\infty$ are given by the same differential expression $\tau = -\partial_r^2 + l(l+1)/r^2$, but, unlike the idealized case, they
are equipped with different boundary conditions (see Definition \ref{A_theta}). 

A point $\la\!\in\!\C$ is an eigenvalue of $\CA_l$ if and only if there exists a $y\!=\!(y_1,y_2)^{\rm t}\! \in\! L_2(0,1)\oplus L_2(0,1)$, $y\ne 0$, with
$y_i, y_i' \!\in {\rm AC_{loc}}((0,1]), \ \tau y_i \in L_2(0,1)$ for $i=1,2$, such that
\begin{align}
\label{tau1}
 (-\tau-\la) y_1 + \al_0 y_2 = 0,\\
\label{tau2}
 \al_0 \tau y_1 + (-\tau-\la) y_2 = 0,
\end{align}
and
\begin{equation} 
\label{tau0}
  y_1'(1) + ly_1(1)=0, \quad y_2(1)=0.
\end{equation}
Clearly, the equations \eqref{tau1}, \eqref{tau2} are equivalent to the relations
\begin{align}
\label{tau4}
  &\big( (\tau + \la)^2 -\al_0^2 \tau \big) y_1 = 0,\\
\label{tau3}
  &\al_0 y_2 = (\tau+\la) y_1.
\end{align}
By Proposition \ref{alpha-const}, we already know that for constant $\al$ all eigenvalues are contained in the interval $(-\infty,\al_0^2/4)$. For $\la \le \al_0^2/4$, we denote by
\[
 v_\pm(\la) := -\la + \frac{\al_0^2}2 \pm \al_0 \sqrt{\frac{\al_0^2}4 - \la} = \bigg( \frac{\al_0}2 \pm \sqrt{\frac{\al_0^2}4-\la}\bigg)^2 \ge 0
\]
the two solutions of the quadratic equation $(v+\la)^2 - \al_0^2 v = 0$. Then we can factorize the differential expression in \eqref{tau4} as
\begin{equation}
\label{factorize}
  (\tau + \la)^2 -\al_0^2 \tau = \big(\tau - v_+(\la) \big) \big(\tau - v_-(\la) \big).
\end{equation}
Since the differential expression $\tau$ is in limit point case at $0$ (see the proof of Pro\-posi\-tion~\ref{A_theta}),
every fundamental system of a differential equation $(\tau - \mu) x = 0$ has exactly one solution in $L_2(0,1)$. This and the relation \eqref{factorize} imply that 
every fundamental system of the differential equation \eqref{tau4} has exactly two solutions in $L_2(0,1)$, which may be chosen as the solutions $y_{1,l}^\pm \in L_2(0,1)$ of 
$\big(\tau - v_\pm(\la) \big) x = 0$ given by (see \eqref{RiccBessel})
\[
 y_{1,l}^\pm(r,\la) = \sqrt{\frac\pi 2} \sqrt{r \,k_\pm(\la)} \, J_{l+1/2}(r \,k_\pm(\la)), \quad r\in[0,1],
\]
with
\begin{equation}
\label{kpm}
  k_\pm(\la):= \sqrt{v_\pm(\la)} = \frac{\al_0}2 \pm \sqrt{\frac{\al_0^2}4-\la}.
\end{equation}
Note that $k_+(\la)+k_-(\la)=\al_0$ and $k_+(\la)k_-(\la)=\la$, that is, $k_\pm(\la)$ are the so\-lutions of the quadratic equation $k^2-\,\al_0 k + \la =0$.
By definition of $v_\pm(\la)$ and $k_\pm(\la)$, we have $(v_\pm(\la)+\la)^2 = \al_0^2 v_\pm(\la) = \al_0^2 k_\pm(\la)^2$. Thus
equation \eqref{tau3} implies~that
\[
 y_{2,l}^\pm(r,\la) = \frac 1{\al_0} (\tau+\la) y_{1,l}^\pm(r,\la)  =  \frac 1{\al_0} \big(v_\pm(\la)+\la \big) y_{1,l}^\pm(r,\la) = k_\pm(\la) y_{1,l}^\pm(r,\la).
\]
Therefore a fundamental system $\{y_l^+,y_l^-\}$ in $L_2(0,1) \oplus L_2(0,1)$ of the system \eqref{tau1}, \eqref{tau2} is given by
\[
  y_l^+(r,\la) = \binom 1{k_+(\la)} y_{1,l}^+(r,\la), \quad y_l^-(r,\la) = \binom 1{k_-(\la)} y_{1,l}^-(r,\la).
\]
If we take into account the boundary conditions and use relation \eqref{Bessel-diff}, we see that there exists a non-zero solution of the boundary eigenvalue problem
\eqref{tau1}, \eqref{tau2}, \eqref{tau0} if and only if
\[
 \left| \begin{array}{ll} 
 k_+(\la) y_{1,l-1}^+(1) & k_-(\la) y_{1,l-1}^-(1) \\ k_+(\la) y_{1,l}^+(1) & k_-(\la) y_{1,l}^-(1)
 \end{array} \right| = 0
\]
or, equivalently, 
\begin{equation}
\label{eig-val-phys}
 J_{l-1/2}\big(k_+(\la)\big) J_{l+1/2}\big(k_-(\la)\big) - J_{l+1/2}\big(k_+(\la)\big) J_{l-1/2}\big(k_-(\la)\big) = 0
\end{equation}
with $k_\pm(\la)$ given by \eqref{kpm}.

This shows that the eigenvalues of the physical dynamo problem \eqref{diffsyst}, \eqref{boundcond} are only given implicitly as the solutions 
of the equation \eqref{eig-val-phys}. As a consequence, even in the simplest case of constant $\al$, it is difficult to obtain any analytic information 
about the eigenvalues of $\CA_l$. 

The estimates in Theorems \ref{!!!} and \ref{strip}, however, provide the global~bound 
\[
  \max \sigma(\CA_l) \le \min \{ a_l,b_l \} = b_l
\]
where 
\[
  b_l = \left\{ \begin{array}{ll}
  \displaystyle{- \frac{\la_1(\infty)\!+\!\la_1(l)}2 \!+\! \sqrt{ \!\bigg( \!\frac{\la_1(\infty) \!-\! \la_1(l)}2\!\bigg)^{\!\!2} \!\!\!+\! \la_1(l) |\al_0|^2\!}}
                                     \ \ & \text{ if } \ |\al_0|^2 \!\le\! \la_1(\infty),  \\[2mm]    
  -\la_1(\infty) + |\al_0|^2  \ \ & \text{ if } \  |\al_0|^2 \!>\! \la_1(\infty).
  \end{array} \right.
\]
In addition, the estimate in Proposition \ref{alpha-const} for constant $\al$ yields the bound
\[
 \max \sigma(\CA_l)  \le \frac{|\al_0|^2}4,
\]
which was already used to establish the eigenvalue relation \eqref{eig-val-phys}.

Moreover, since $\CB_l=\CS_l$ for constant $\al$, Proposition \ref{july6} yields that the number $k_0$ of positive eigenvalues of $\CA_l$ coincides with the number of positive eigenvalues of the operator $-A_\infty + \al_0^2$ or, equivalently, %with the number of eigenvalues of the Bessel operator $A_\infty$ less than $\al_0^2$:
\[
  \# \,\big( \sigma(\CA_l) \cap (0,\infty) \big) = \# \,\big( \sigma(A_\infty) \cap (0,\al_0^2) \big);
\]
if we enumerate the eigenvalues of $\CA_l$ and of $A_\infty$ as $0< \la_{k_0}(\CA_l) \le \dots \le \la_1(\CA_l)$ and $\la_1(\infty) \le \dots \le \la_{k_0}(\infty) < |\al_0|^2$, respectively, 
we have the estimate
\[
  0 < \la_j(\CA_l) \le -\la_j(\infty) + |\al_0|^2, \quad j=1,2,\dots, k_0.
\]
\end{example}

Finally, we consider the physical dynamo problem with non-constant helical turbulence function $\al$. 
As shown for the first time in \cite{stefani-2003-67}, there exist special $\alpha$-profiles which provide  dipole-dominated criticality 
for oscillatory dynamo regimes (that is, non-real eigenvalues passing from the left to the right half-plane first for dipole modes ($l=1$) and later for quadrupole and higher-degree modes ($l>1$)). 
Such regimes are of high physical interest due to their close relation to polarity reversal processes of the magnetic field (see \cite{stefani:184506}, \cite{2006E&PSL.243..828S}).

\begin{example}
\label{frank}
Consider the physical dynamo problem \eqref{diffsyst}, \eqref{boundcond} for $l=1$ with $\alpha$ given by
\begin{equation}
\label{frank-alpha}
  \al(r)=C(-21.46+426.41\,r^2-806.73\,r^3+392.28 r^4), \quad r\in[0,1],
\end{equation}
where $C\ge 0$ is a constant. For $C=0$ the eigenvalues are  real and they coincide with the 
interlacing sequences of eigenvalues of the Bessel operators $-A_\infty$ and $-A_l$. According to the numerical computations in
\cite{stefani-2003-67}, if $C$ increases, the largest two eigenvalues merge at $C=0.818$ and form a complex conjugate pair for $0.818 < C < 1.097$.
In between, at $C=1$ this pair crosses the imaginary axis.

According to the abstract result of Corollary \ref{meet}, the largest two eigenvalues can only meet if condition \eqref{meet-cond} is violated. 
Indeed, for $l=1$, we have
\[
  \la_1(l) = \pi^2 \approx 9.87, \quad \la_1(\infty) \approx 20.19.
\]
and, for $\al$ given by \eqref{frank-alpha} with $C=0.818$, 
\[
 \|\al\| = \max_{r\in[0,1]} |\al(r)| \approx 17.55, \quad \|\al'\| \approx \max_{r\in[0,1]} |\al'(r)| \approx 71.36,
\]
and hence condition \eqref{meet-cond} does not hold since
\[
 \|\al\|^2 + \frac{\|\al\|\,\|\al'\|}{\sqrt{\la_1(\infty)}} \approx 586.72 > 1.32 \approx \frac {\big(\la_1(\infty)-\la_1(\theta) \big)^2}{4\la_1(\infty)}.
\]
\end{example}

} % end of newcommmand \i

\bibliographystyle{alpha-abbrv}
%\nocite{*}
\bibliography{dynamo}

\end{document}